%% file: main.tex
\documentclass[sigplan, screen]{acmart}

\acmConference[PL'18]{ACM SIGPLAN Conference on Programming Languages}{January 01--03, 2018}{New York, NY, USA}
\acmYear{2018}
\acmISBN{} 
\acmDOI{} 
\startPage{1}

\setcopyright{none}

\usepackage{microtype}
\usepackage{enumitem}
\usepackage{amssymb}
\usepackage{amsmath}
\usepackage{bbm}
\usepackage[plain,figure]{algorithm}
\usepackage[noend]{algpseudocode}
\usepackage{caption}
\usepackage{multirow}

\usepackage{tikz}
\usetikzlibrary{positioning,shapes,arrows}
\usepackage{pgfplots}
\usetikzlibrary{plotmarks}

\usepackage{xcolor,colortbl}

\usepackage{soul}

\definecolor{antiquefuchsia}{RGB}{145,92,131}
\definecolor{burnishbrown}{RGB}{161,122,116}
\definecolor{bluegray}{RGB}{102,153,204}

\definecolor{mediumbg}{RGB}{60,91,112}

\definecolor{darkbg}{RGB}{113,70,35}
\definecolor{mediumbg}{RGB}{79,117,139}
\definecolor{lightbg}{RGB}{173,164,174}

\newcommand{\revise}[1]{\textcolor{black}{#1}}

\usepackage{booktabs}   
\usepackage{subcaption} 

\usepackage[toc,page]{appendix}
\usepackage[bottom]{footmisc}

\begin{document}

\title[Short Title]{Multi-modal Synthesis of Regular Expressions}         


\author{Qiaochu Chen}
\affiliation{
  \institution{University of Texas at Austin}            
  \city{Austin}
  \state{Texas}
  \country{USA}                    
}
\email{qchen@cs.utexas.edu}          

\author{Xinyu Wang}
\affiliation{
  \institution{University of Michigan, Ann Arbor}            
  \city{Ann Arbor}
  \state{Michigan}
  \country{USA}                    
}
\email{xwangsd@umich.edu}         

\author{Xi Ye}
\affiliation{
  \institution{University of Texas at Austin}            
  \city{Austin}
  \state{Texas}
  \country{USA}                    
}
\email{xiye@cs.utexas.edu}          

\author{Greg Durrett}
\affiliation{
  \institution{University of Texas at Austin}            
  \city{Austin}
  \state{Texas}
  \country{USA}                    
}
\email{gdurrett@cs.utexas.edu}          

\author{Isil Dillig}
\affiliation{
  \institution{University of Texas at Austin}            
  \city{Austin}
  \state{Texas}
  \country{USA}                    
}
\email{isil@cs.utexas.edu}          

\input{macros}

\input{abstract}



\maketitle

\input{intro}
\input{overview}
\input{top_level_alg}

\input{sketch_completion}

\input{sketch_generation}

\input{impl}
\input{setup}

\input{eval}
\input{related}

\input{conc}


\bibliography{main}

\input{appendix}

\end{document}

%% file: macros.tex
\newcommand{\sentence}{{\mathcal{L}}}
\newcommand{\examples}{\mathcal{E}}
\newcommand{\posexamples}{\mathcal{E}^{+}}
\newcommand{\negexamples}{\mathcal{E}^{-}}
\newcommand{\posexample}{e^{+}}
\newcommand{\negexample}{e^{-}}
\newcommand{\matches}{\models}
\newcommand{\annot}{\triangleleft}
\newcommand{\hsketch}{\mathcal{S}}
\newcommand{\program}{\mathcal{P}}
\newcommand{\nlabel}{\ell}
\newcommand{\symconst}{\kappa}

\newcommand{\toolname}{{\sc Regel}}
\newcommand{\toolvar}{\toolname-{\sc Pbe}}
\newcommand{\overapprox}{o}
\newcommand{\underapprox}{u}
\newcommand{\exps}{\Pi}

\newcommand{\gensketch}{\textsc{GenSketch}}
\newcommand{\completesketch}{\textsc{Synthesize}}

\algnewcommand\Input{\textbf{input: }}
\algnewcommand\Output{\textbf{output: }}

\newcommand{\assign}{:=}
\newcommand{\prog}{r}
\newcommand{\partialprog}{P}
\newcommand{\sketch}{{\mathcal{S}}}

\newcommand{\worklist}{\emph{worklist}}

\newcommand{\hole}{\square}

\newcommand{\unknown}{\texttt{?}}

\renewcommand{\dots}{\cdot\cdot}
\renewcommand{\ldots}{\dots}

\newcommand{\firstBL}{\toolname-{\sc Enum}}
\newcommand{\secondBL}{\toolname-{\sc Approx}}
\newcommand{\mina}{{\sc AlphaRegex}}

\newcommand{\todo}[1]{{\color{red}{#1}}}

\newcommand{\semantics}[1]{\llbracket{#1}\rrbracket}
\newcommand{\bigsemantics}[1]{\big\llbracket{#1}\big\rrbracket}
\newcommand{\Bigsemantics}[1]{\Big\llbracket{#1}\Big\rrbracket}

\renewcommand{\algref}[1]{Algorithm~\ref{alg:#1}}
\newcommand{\alglabel}[1]{\label{alg:#1}}
\newcommand{\figref}[1]{Figure~\ref{fig:#1}}
\newcommand{\figlabel}[1]{\label{fig:#1}}
\newcommand{\exref}[1]{Example~\ref{exp:#1}}
\newcommand{\exlabel}[1]{\label{exp:#1}}
\newcommand{\defref}[1]{Definition~\ref{def:#1}}
\newcommand{\deflabel}[1]{\label{def:#1}}
\newcommand{\tabref}[1]{Table~\ref{tab:#1}}
\newcommand{\tablabel}[1]{\label{tab:#1}}
\newcommand{\lemmaref}[1]{Lemma~\ref{lem:#1}}
\newcommand{\lemmalabel}[1]{\label{lem:#1}}
\newcommand{\secref}[1]{Section~\ref{sec:#1}}
\newcommand{\seclabel}[1]{\label{sec:#1}}
\newcommand{\theoremref}[1]{Theorem~\ref{thm:#1}}
\newcommand{\theoremlabel}[1]{\label{thm:#1}}

\newcommand{\deepregex}{{\sc DeepRegex}}

\newcommand{\irule}[2]{\mkern-2mu\displaystyle\frac{#1}{\vphantom{,}#2}\mkern-2mu}

%% file: abstract.tex
\begin{abstract}
In this paper, we propose a multi-modal synthesis technique for automatically constructing regular expressions (\emph{regexes}) from a combination of examples and natural language. Using multiple modalities  is useful in this context because natural language alone is often highly ambiguous, whereas examples in isolation are often not sufficient for conveying  user intent. Our proposed technique first parses the English description into a so-called \emph{hierarchical sketch} that guides our programming-by-example (PBE) engine. Since the hierarchical sketch captures crucial hints, the PBE engine can leverage this information to both prioritize the search as well as make useful deductions for pruning the search space.

We have implemented the proposed technique in a tool called \toolname\ and  evaluate it on over three hundred regexes.  Our evaluation shows that \toolname\  achieves 80\% accuracy whereas the NLP-only and PBE-only baselines achieve  43\% and 26\% respectively.  We also compare our proposed PBE engine against an adaptation of \mina, a state-of-the-art regex synthesis tool, and show that our proposed PBE engine is an order of magnitude faster, even if we adapt the search algorithm of \mina\ to leverage the sketch. Finally, we conduct a user study involving 20 participants and show that users are twice as likely to successfully come up with the desired regex using \toolname\ compared to without  it.

\end{abstract}

%% file: intro.tex
\section{Introduction} \label{sec:intro}

As a convenient mechanism for matching patterns in text data, regular expressions (or \emph{regexes}, for short) have found numerous applications ranging from search and replacement to input validation. In addition to being heavily used by programmers, regular expressions have also gained popularity among computer end-users. For example, many text editors, word processing programs, and spreadsheet applications now provide support for performing search and replacement using regexes.  However, despite their potential to dramatically simplify various tasks, regular expressions have a reputation for being quite difficult to master.

Due to the practical importance of regexes, prior research has proposed methods to automatically generate regular expressions from high-level user guidance. For example, several techniques generate regexes from natural language descriptions~\cite{KB13,deepregex,semregex}, while others  synthesize regexes from positive and negative examples~\cite{flashfill, fidex, lee}.
While these techniques have made some headway in regex synthesis, 
existing NLP-based techniques have relatively low accuracy even for stylized English descriptions~\cite{deepregex},
whereas example-based synthesizers  impose severe restrictions  on  the kinds of regular expressions they can synthesize (e.g., {restrict} the use of Kleene star~\cite{fidex,flashfill} or consider only a binary alphabet~\cite{lee}).

A central premise of this work is that both modalities of information, namely examples and natural language, are complementary and simultaneously useful for synthesizing regular expressions. As evidenced by numerous regex-related questions posted on online forums, most users communicate their intent using a combination of natural language and positive/negative examples. In particular, a common pattern is that users typically describe the high-level task using natural language,  but they also give positive and negative examples to clarify any ambiguities present in that description.

Motivated by this observation, this paper presents a new \emph{multi-modal} synthesis algorithm that utilizes both examples and English text to generate the target regex.  
The key idea underlying our method is to  parse the natural language description into a \emph{hierarchical sketch} (or \emph{h-sketch} for short) that is  used to guide a programming-by-example (PBE) engine. Since hierarchical sketches capture key hints present in the English description, they make it much easier for our PBE technique to find regexes that match the user's intent. Furthermore, because the hierarchical nature of these sketches closely reflects the compositional structure of the natural language they are derived from, it is feasible to  obtain the basic scaffolding of the target regex using non-data-hungry NLP techniques like \emph{semantic parsing}~\cite{mooney,zettlemoyer}.

In order to effectively use the hints derived from natural language, our technique leverages a new PBE algorithm for the regex domain. In particular, our  PBE technique uses the hints provided by the h-sketch to both prioritize its search and also perform useful deductive reasoning. In addition, our PBE technique  leverages so-called \emph{symbolic regular expressions} to group similar regexes during the search process and uses an SMT solver to concretize them.

We have implemented the proposed approach in a tool called \toolname\footnote{Stands for Regular Expression Generation from Examples and Language.  }
and compare it against relevant baselines  on \emph{over 300 regexes} collected from two different sources. Our evaluation demonstrates 
 the advantages of multi-modal synthesis compared to both \deepregex, a state-of-the-art NLP tool, as well as a pure PBE approach. In particular, \toolname\ can find the intended regex in 80\% of the cases, whereas the pure PBE and NLP baselines can synthesize only 26\% and 43\% of the benchmarks respectively. In our evaluation, we also compare \toolname's PBE engine  against an adaptation of \mina, a state-of-the-art PBE tool for regex synthesis, and demonstrate an order of magnitude improvement in terms of sketch completion time.  Finally,
we perform a user study targeting real-world regex construction tasks  and show that users are twice as likely to construct the intended regex using \toolname\ than without it.



In summary, this paper makes the following  contributions:
\begin{itemize}[leftmargin=*,noitemsep,topsep=0pt]
\item We describe a new multi-modal synthesis technique for generating regexes from examples and natural language.
\item We introduce \emph{hierarchical sketches}  and develop a semantic parser to  generate h-sketches from English descriptions.
\item We present a new PBE engine for regular expression synthesis that (1) leverages hints in the h-sketch to guide both the search and deduction, and (2) utilizes the concept of \emph{symbolic regexes} to further prune the search space.
\item We evaluate our technique on over 300 regexes and  empirically demonstrate its advantages against multiple baselines on  two different data sets. 
\item We conduct a user study and run statistical significance tests to evaluate  the benefits of \toolname\ to prospective users. 
\end{itemize}

%% file: overview.tex
\section{Overview}\label{sec:overview}

\begin{figure*}
\begin{center}
\begin{minipage}{0.25\textwidth}
\begin{center}
\includegraphics[scale=0.5]{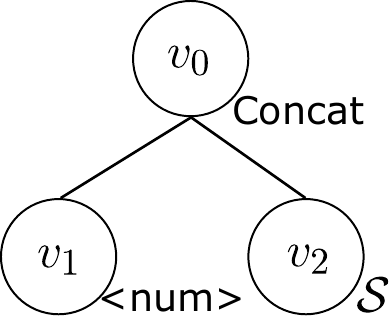}
\end{center}
\vspace{15pt}
\captionsetup{font={scriptsize}}
\caption{\scriptsize A partial regex example where $\sketch$ represents the h-sketch $\hole_2\{\texttt{<,>},\texttt{RepeatRange}(\texttt{<num>,}\texttt{1,3})\}$. 
}\label{fig:partial}
\end{minipage}
\hspace{10pt}
\begin{minipage}{0.35\textwidth}
\begin{center}
\includegraphics[scale=0.5]{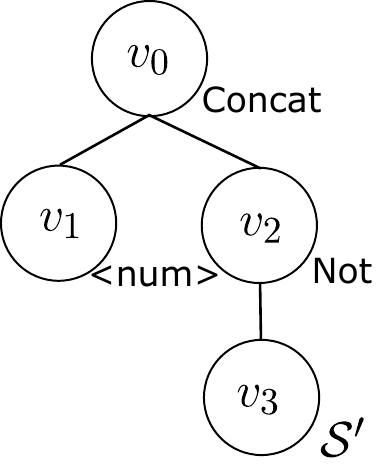}
\end{center}
\captionsetup{font={scriptsize}}
\caption{\scriptsize A partial regex expanded from Figure~\ref{fig:partial} where $\sketch'$ stands for 
 $\hole_1\{\texttt{<,>},\texttt{RepeatRange}(\texttt{<num>,1,3})\}$
 . 
}\label{fig:partial-expanded}
\end{minipage}
\hspace{10pt}
\begin{minipage}{0.3\textwidth}
\begin{center}
\includegraphics[scale=0.5]{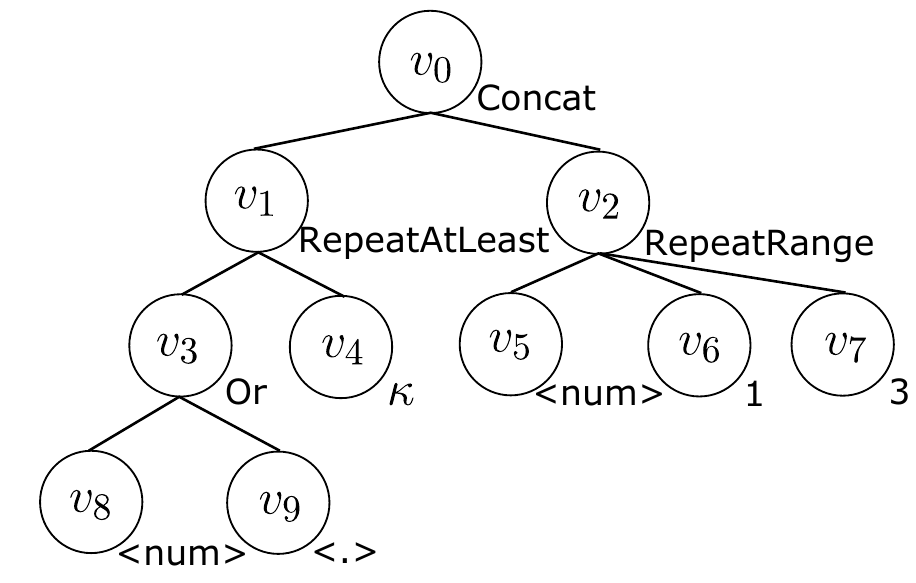}
\end{center}
\captionsetup{font={scriptsize}}
\caption{\scriptsize A symbolic regex example.}\label{fig:symbolic}
\end{minipage}
\end{center}
\end{figure*}

In this section, we give a high-level overview of our technique with the aid of a motivating example. Consider the task of writing a regular expression to match strings that correspond to decimal numbers of the form  $x.y$ where $x$ (resp. $y$) is an integer with at most $15$ (resp. $3$) digits. Furthermore, this regex should  accept strings that correspond to $15$ digit integers (i.e., where the $.y$ part is missing).

As posted in a StackOverflow post,\footnote{
https://stackoverflow.com/questions/19076566/need-regular-expression-that-validate-decimal-18-3} the user  explains this task using the following English description $\mathcal{L}$:  
\emph{
        ``I need a regular expression that validates Decimal(18, 3), which means the max number of digits before comma is 15 then accept at max 3 numbers after the comma.''
}
The user also provides some positive examples $\mathcal{E}^+$ and negative examples~$\mathcal{E}^-$: \\

\begin{center}
\vspace{-0.1in}
\small
\begin{tabular}{|c|c|}
\hline
{\bf Positive Examples} & {\bf Negative Examples}\\ \hline
123456789.123 & 1234567891234567 \\
123456789123456.12 & 123.1234 \\
12345.1 & 1.12345 \\
123456789123456 & .1234 \\
\hline
\end{tabular}
\end{center} 
\vspace{0.1in}

Here, the user's English description is not only ambiguous, but also somewhat misleading. First, the user means to say ``period'' instead of ``comma'', and, second, it is not clear from the description whether a pure integer such as ``123'' should be allowed. On the other hand, the string examples alone are also not sufficient for completely understanding user intent. For instance, by looking at the examples in isolation, it is  difficult to tell whether digit 0 is allowed or not. 
\

To synthesize the target regex based on the user's description and examples, our method first uses a semantic parser~\cite{berant} to ``translate'' the natural language description into a \emph{hierarchical sketch (h-sketch)} that  captures the high-level structure of the target regex. Given the English description $\mathcal{L}$, our semantic parser generates a \emph{ranked} list of such h-sketches, one of which is given below:
\small
\begin{equation}\label{eq:hsketch}
\resizebox{.92\hsize}{!}{
$\texttt{Concat} \Big( \hole{\{\texttt{<num>},\texttt{<,>}\}},\hole\{\texttt{RepeatRange}(\texttt{<num>,1,3}), \texttt{<,>}\} \Big)$
}
\end{equation}
\normalsize

In this h-sketch, the symbol $\hole$ denotes an unknown regex, and the notation $\hole\{\sketch_1, \ldots, \sketch_n\}$ indicates that the unknown regex  $\hole$ should contain \emph{at least} one of the components (``hints'') $\sketch_1, \ldots, \sketch_n$ as a leaf node. Thus, looking at this h-sketch, we can make the following observations:
\begin{enumerate}[leftmargin=*]
\item Since the top-level operator is {\tt Concat}, the regular expression is of the form $\texttt{Concat}(R_1, R_2)$. 
\item $R_1$ should contain \emph{either} a digit (i.e., {\tt <num>}) \emph{or} a comma (i.e., {\tt <,>}) as a component. 
\item $R_2$ should contain \emph{either} a 1-3 digit number (i.e.,\\ {\tt RepeatRange(<num>,1,3)}) \emph{or} a comma. 
\end{enumerate}
While this sketch is far from perfect,  it still contains useful sub-regexes  that do indeed occur in the target regex.

Given a hierarchical sketch $\hsketch$ like the one from Eq.~\ref{eq:hsketch}, our PBE engine tries to find a regex that is both a valid completion of $\hsketch$ and also consistent with the provided  examples. From a high level, the synthesizer performs top-down sketch-guided enumerative search over \emph{partial regexes} represented as abstract syntax trees (ASTs). For instance, Figure~\ref{fig:partial} shows an example partial regex where nodes are  labeled  with h-sketches, operators, or character classes. At every step, the synthesizer picks a node labeled with a sketch and decides how to expand that node. For instance, Figure~\ref{fig:partial-expanded} shows an expansion of the partial regex from Figure~\ref{fig:partial} where the node $v_2$ has been instantiated with the \texttt{Not} operator which now has a new child $v_3$ labeled with a new h-sketch $\hsketch'$. 
~\footnote{In \figref{partial} and ~\figref{partial-expanded}, the notation $\hole_k$ indicates that the depth of the unknown regex is at most $k$. Thus, when we derive the new sketch for node $v_3$, we use the same sketch labeling $v_2$ but with depth $1$ instead of $2$. }

The synthesis engine underlying \toolname\ leverages two  ideas that help make it practical. 
{First, similar to prior work~\cite{lee}, \toolname\ uses lightweight deductive reasoning to prune away infeasible partial regexes by constructing over- and under-approximations. 
However, with our h-sketches, we are able to construct these approximations using hints obtained from the natural language and therefore perform more precise reasoning.}
Specifically, given a partial regex $\partialprog$, our PBE engine uses  the h-sketch to construct a pair of regular expressions $\langle o, u \rangle$ such that 
(1) $o$ accepts every string that \emph{any} completion of $\partialprog$ can match, and 
(2) $u$ accepts only those strings that \emph{every} completion of $\partialprog$ accepts. 
For instance, the under-approximation for the partial regex from Figure~\ref{fig:partial-expanded}~is:
\small
\begin{equation}\label{eq:underapprox}
\resizebox{.9\hsize}{!}{
$\texttt{Concat} \Big( \texttt{<num>},\texttt{Not} \big( \texttt{Or}(\texttt{<,>},\texttt{RepeatRange}(\texttt{<num>,1,3})) \big) \Big)$
}
\end{equation}
\normalsize
Since this regex recognizes the negative example $``123456789\\12345467"$, \emph{any} completion of the partial regex from Figure~\ref{fig:partial-expanded} must also recognize this negative example. Thus, we can reject this partial regex without compromising completeness.



The second  idea underlying our synthesis algorithm is to introduce \emph{symbolic regexes}  to prune large parts of the search space. In particular, our regex DSL has several constructs (e.g., \texttt{RepeatRange}) that take integer constants as arguments, but explicitly enumerating possible values of these integer constants during synthesis can be quite inefficient. To deal with this challenge,
our  algorithm  introduces a so-called \emph{symbolic integer} $\kappa$ that represents \emph{any} integer value. Now, given a \emph{symbolic regex} with symbolic integers, our method  generates an SMT formula $\phi$ over the symbolic integers $\kappa_1, \ldots, \kappa_n$ such that $\kappa_i$ can be instantiated with constant $c_i$ only if $c_1, \ldots, c_n$ is a model of $\phi$. For instance, consider the symbolic regex  from Figure~\ref{fig:symbolic}. By looking at each of the sub-regexes of Figure~\ref{fig:symbolic}, we can make the following deductions:

\begin{itemize}[leftmargin=*]
\item Since $v_3$'s arguments (an {\tt Or} node) are both single characters, any string matched by $v_3$ must have length $1$.
\item Because  {\tt RepeatAtLeast}  concatenates at least $\kappa$ copies of its first argument, the length of any string matched by $v_1$ is at least $\kappa$.
\item Finally, the length of any string matched by $v_0$ must be at least $\kappa+1$ because $v_0$'s first (resp. second) argument has length at least $\kappa$ (resp. $1$).
\end{itemize}

Now, since there is  a positive example (namely, $12345.1$) of length $7$, this gives us the  constraint  $\kappa + 1 \leq 7$  (i.e., $\kappa \leq 6$) on the symbolic integer $\kappa$. Thus, rather than enumerating all possible integers, our approach instead generates an SMT formula and solves for possible values of the symbolic integers. However, because the generated SMT formula $\Phi$  over-approximates ---rather than precisely encodes---  regex semantics, not every model of $\phi$ corresponds to a regex that is consistent with the examples. Thus, our approach uses SMT solving to prune infeasible symbolic regexes rather than directly solving for the unknown constants (e.g., as is done in \textsc{Sketch}~\cite{solarthesis} and its variants~\cite{jha2010, gulwani2011, ahish2015, emina}).

Using these ideas, our synthesis algorithm is able to  synthesize the following correct regex:
\small
\begin{align*}
&\texttt{Concat} \Big( \texttt{RepeatRange}(\texttt{<num>,1,15}),\\[-0.8em]
& \ \ \ \ \ \ \ \ \ \ \ \ \ \ \texttt{Optional} \big( \texttt{Concat}(\texttt{<.>},\texttt{RepeatRange}(\texttt{<num>,1,3})) \big) \Big) 
\end{align*}
\normalsize


%% file: top_level_alg.tex


\section{Regex Language} 

\begin{figure}
\[
\small
\begin{array}{lll}
\prog  := \ c \ | \ \epsilon \ | \ \emptyset \ \\ 
    \ \ \ \ \ | \ \texttt{StartsWith} ( \prog ) \ | \ \texttt{EndsWith} ( \prog ) \ | \  \texttt{Contains} ( \prog ) \ | \ \texttt{Not} ( \prog )  \\  
    \ \ \ \ \ | \ \texttt{Optional} (\prog) \ | \ \texttt{KleeneStar} ( \prog )  \\
    \ \ \ \ \ | \ \texttt{Concat} ( \prog_1, \prog_2 ) \ | \  \texttt{Or} ( \prog_1, \prog_2 ) \ | \  \texttt{And} (\prog_1, \prog_2)  \\ 
    \ \ \ \ \ | \ \texttt{Repeat} ( \prog, k ) \ | \  \texttt{RepeatAtLeast} (\prog, k) \ | \ \texttt{RepeatRange} ( \prog, k_1, k_2 ) \\ 
\end{array}
\]
\vspace{-5pt}
\caption{Regex DSL. Here, $k \in \mathbb{Z}^+$ and $c$ is a character class}
\figlabel{dslsyntax}
\end{figure}

Following prior work~\cite{deepregex}, we express regular expressions in the simple DSL shown in \figref{dslsyntax}.~\footnote{The precise semantics of this DSL are provided in the Appendix.}  While most constructs in this DSL are just syntactic sugar for standard regular expressions, the  \texttt{And} and \texttt{Not} operators may require performing intersection and complement at the automaton level. However, any ``program'' in our DSL is expressible as a standard regex, and, furthermore, several regex libraries~\cite{pattern,ruby} already directly support some forms of \texttt{And} and \texttt{Not}. In what follows, we briefly go over the regex constructs shown in  \figref{dslsyntax}.

\paragraph{Character class.} A character class $c$ is either a single character (e.g., {\tt <a>}, {\tt <1>}, {\tt <,>}) or a predefined family of characters. For instance, the character class {\tt <num>} matches any digit \texttt{[0-9]},  {\tt <let>} matches any letter {\tt [a-zA-Z]}, and {\tt <cap>} and {\tt <low>} match upper and lower case letters respectively.  We also have a character class {\tt <any>} that matches any character, {\tt <alphanum>} matches alphanumeric characters, and {\tt <hex>} matches hexadecimal characters.  

\paragraph{Containment.} The DSL operator {\tt StartsWith}($\prog$) (resp. \\{\tt EndsWith}($\prog$)) evaluates to true on string $s$ if there is a prefix (resp. suffix) of $s$ that matches  $\prog$. Similarly, {\tt Contains}($\prog$) evaluates to true on $s$ if any substring of $s$ matches $\prog$.

\paragraph{Concatenation.} The operator {\tt Concat}($r_1, r_2$) evaluates to true on string $s$ if $s$ is a  concatenation of two strings $s_1, s_2$ that match $r_1, r_2$ respectively.

\paragraph{Logical operators.} The operator {\tt Not}($r$) matches a string $s$ if  $s$ does not match $r$.  Similarly,  {\tt And}($r_1, r_2$) (resp. {\tt Or}($r_1, r_2$)) matches $s$ if $s$ matches both (resp. either) $s_1$ and (resp. or) $s_2$. The  construct {\tt Optional}($r$) is syntactic sugar for ${\tt Or}(\epsilon, r)$.

\paragraph{Repetition.} The construct {\tt Repeat}($r, k$) matches string $s$ if $s$ is a concatenation of exactly $k$ strings $s_1, \ldots, s_k$ where each $s_i$ matches $r$. {\tt RepeatRange}($r, k_1, k_2$) matches string $s$ if there exists some $k \in [k_1, k_2]$ such that  {\tt Repeat}($r, k$)  matches $s$. 
Finally, {\tt RepeatAtLeast}($r, k$) is just syntactic sugar for {\tt RepeatRange}($r, k, \infty$), and {\tt KleeneStar}($r$) is equivalent to {\tt Or}($\epsilon$, {\tt RepeatAtLeast}($r, {\tt 1}$))). 
Note that  operators in the {\tt Repeat} family require every integer value $k$ to be a positive number.

\section{Hierarchical Sketches}
\seclabel{sketchlang}

\begin{figure}
    \[
    \small
    \begin{array}{lll}
    \sketch & :=  \ \hole_d \{ \overline{\sketch} \} & \text{(constrained hole)}   \\  
    & \ \  | \ \ \ \texttt{f} (\overline{\sketch}) & \text{(operator without symbolic integer)}  \\ 
    & \ \  | \ \ \ \texttt{g} (\sketch, \overline{\kappa}) & \text{(operator with symbolic integer)} \\ 
    & \ \  | \ \ \   r & \text{(regex)} \\ 
    \end{array}
    \]
    \vspace{-10pt}
    \caption{Syntax of hierarchical sketch language where $\prog$ is a concrete regex and $\kappa_i$ is a symbolic integer.}
    \figlabel{sketchsyntax} 
    \end{figure}

In this section, we  present the syntax  and semantics  of hierarchical sketches (h-sketches) that we derive from the natural language.
Intuitively, an h-sketch represents a \emph{family} of regexes that  conform to a high-level structure.

As shown in \figref{sketchsyntax}, our h-sketch language extends our regex DSL by allowing a  ``constrained hole'' construct.
A constrained hole, denoted $\hole_{d} \{ \overline{\sketch} \}$, is an unknown regex that is parametrized with a positive integer $d$ and a set of nested h-sketches $\overline{\sketch}$. 
Specifically, regex $\prog$ belongs to the space of regexes defined by $\hole_{d} \{ \overline{\sketch}\}$ if one of the ``leaf'' nodes of $\prog$ conforms to $\sketch_i$ and $\prog$ has depth at most $d$ (when  $\sketch_i$ is viewed as a ``leaf node''). Observe that constrained holes can be arbitrarily nested, which is why  these sketches are \emph{hierarchical}.


In addition to constrained holes,   h-sketches can also contain operators in our regex DSL. 
For example, an h-sketch can be of the form $\texttt{f}(\overline{\sketch})$ where $\texttt{f}$ denotes a DSL operator    outside of the \texttt{Repeat} family (e.g., $\texttt{Concat}$). Semantically, $\texttt{f}(\overline{\sketch})$ represents the set of  regexes $\texttt{f}(\overline{\prog})$ where we have $\prog_i \in \semantics{\sketch_i}$. Our h-sketches can also be of the form $\texttt{g}(\sketch, \overline{\kappa})$ where \texttt{g} is a construct in the \texttt{Repeat} family and $\kappa$'s are so-called \emph{symbolic integers}.  
The set of programs defined by $\texttt{g}(\sketch, \overline{\kappa})$ includes all programs of the form $\texttt{g}(\prog, \overline{k})$ where we have $\prog \in \semantics{\sketch}$ and $k_i$ is any positive integer.  
Finally, our h-sketch language also includes \emph{concrete} regular expressions (without holes), and the semantics provided in ~\figref{sketchsemantics} summarize this discussion.

\begin{figure}
    \footnotesize 
\[
\begin{array}{l}
\vspace{5pt}
\bigsemantics{\prog} = \{ \prog \} \\
\vspace{5pt}
\bigsemantics{ \texttt{f} (\overline{\sketch})} = \big\{ \  \texttt{f} (\overline{\prog})  \ | \ \forall_{i \in |\overline{\sketch}|} \ \prog_i \in \semantics{\sketch_i} \  \big\} 
\\ 
\vspace{5pt}
\bigsemantics{ \texttt{g} (\sketch, \overline{\kappa})} = \big\{ \  \texttt{g} (\prog, \overline{k})  \ | \ \prog \in \semantics{\sketch}, \forall_{i \in |\overline{\kappa}|} \  k_i \in \mathbb{N}   \  \big\} 
\\  
\vspace{10pt}
\bigsemantics{\hole_d \{ \overline{\sketch}\}} =  \left\{ 
\begin{array}{lr}
\hspace{-5pt}
\begin{array}{l}
\bigcup\limits_{i \in |\overline{\sketch}|} \ \semantics{\sketch_i} \hspace{50pt}  d = 1  
\end{array}
\\ \\ 
\hspace{-5pt}
\begin{array}{ll}
\hspace{12pt} 
\vspace{5pt}
\bigcup\limits_{i \in |\overline{\sketch}|} \ \semantics{\sketch_i} \hspace{37pt} d > 1 
\\ 
\vspace{-5pt}
\cup 
\bigcup\limits_{\texttt{f} \in \mathcal{F}_n} \bigcup\limits_{1 \leq i \leq n } 
\bigsemantics{ \texttt{f} ( \hspace{-3pt} \underbrace{l, \dots, l}_{i-1 \text{ times}} \hspace{-3pt} , \hole_{d-1}\{\overline{\sketch}\}, \hspace{-5pt} \underbrace{l, \dots, l}_{n-i \text{ times}} \hspace{-4pt} )}\\\\ \hspace{50pt}  \text{where } l = \hole_{d-1} \mathcal{C}  \cup \{\overline{\sketch}\} \\ 
\\ 
\cup 
\bigcup\limits_{ \texttt{g} \in \mathcal{G}_{n}} 
\bigsemantics{ \texttt{g} ( \hspace{1pt} \hole_{d-1}\{\overline{\sketch}\}, \overline{\kappa} \hspace{1pt} )} 
\end{array}
\end{array}
\right.
\end{array}
\]
\vspace{-15pt}
\caption{Semantics of h-sketches.  $\texttt{g} \in \mathcal{G}_n$  (resp. $\texttt{f} \in \mathcal{F}_n$) is an n-ary operator in (resp. outside of) the \texttt{Repeat} family.  }
\figlabel{sketchsemantics} 
\end{figure}

\begin{example}
The program $\texttt{Concat}(\texttt{<num>}, \texttt{Contains}(\texttt{<,>}))$ is in the language of the h-sketch $\texttt{Concat} \big( \hole_1{\{\texttt{<,>}, \texttt{<num>}\}},\\\hole_2\{\texttt{<,>},\texttt{RepeatRange}(\texttt{<num>,1,3})\} \big)$.
\end{example}

\paragraph{Remark.} While  constrained holes in ~\figref{sketchsyntax} are explicitly parametrized by an integer $d$ to facilitate defining h-sketch semantics, the sketches produced by our semantic parser do not have this explicit integer $d$. Instead, $d$ should be thought of as a configurable parameter  that determines the depth of the search tree explored by the PBE engine.

%% file: sketch_completion.tex
\section{Regex Synthesis from  H-Sketches}
\seclabel{synthesis}

In this section, we describe our synthesis algorithm that generates a regex  from an h-sketch $\sketch$ and a set of positive and negative examples, $\posexamples$ and  $\negexamples$. The output of the synthesis procedure is either $\bot$ which indicates an unsuccessful synthesis attempt or a regex $\prog$ such that:
\[
  \hspace{-5pt}
  \small
\begin{array}{lll}
(1) \ \prog \in \semantics{\sketch} &
(2) \  \forall s \in \posexamples. \  \semantics{\prog}_s = \emph{true}  &
(3) \ \forall s \in \negexamples. \ \semantics{\prog}_s = \emph{false} 
\end{array}
\]

Our synthesis procedure is given in ~\figref{sketchcompletion}. At a high-level,  $\completesketch$  maintains a worklist of \emph{partial regexes} and keeps growing this worklist by expanding the \emph{abstract syntax tree} (AST) representation of a partial regex.

  \begin{definition}{\bf (Partial regex)} A \emph{partial regex} $\partialprog$ is a  tree $(V, E, A)$ where $V$ is a set of vertices, $E$ is a set of directed edges, and $A$ is a mapping from each node $v \in V$ to a label $\nlabel$, which is either (1) a DSL construct (e.g., character class or operator), (2) a symbolic integer $\symconst$, or (3) a hierarchical sketch $\hsketch$. 
\end{definition}


In the remainder of this section, we use the term \emph{symbolic regex} to denote a partial regex where all of the node labels are either DSL constructs or symbolic integers (not an h-sketch), and we use the term \emph{concrete regex} to denote a partial regex where all node labels are DSL constructs. Thus, every concrete regex corresponds to a program written in the regex DSL from ~\figref{dslsyntax}. Given a partial regex $\partialprog$, we write $\textsf{IsConcrete}(\partialprog)$ to denote that $\partialprog$ is a concrete regex and $\textsf{IsSymbolic}(\partialprog)$ to indicate that $\partialprog$ is a symbolic (but not concrete) regex. Finally, we refer to any node whose corresponding label is an h-sketch as an \emph{open node}.

\begin{example}
The partial regex shown in Figure~\ref{fig:symbolic} is a symbolic (but not concrete) regex. On the other hand, the partial regexes from Figures~\ref{fig:partial} and ~\ref{fig:partial-expanded} are neither symbolic nor concrete because the nodes labeled with $\sketch$ are \emph{open}.
\end{example}

\paragraph{Notation.} Given a partial regex $\partialprog$ represented as an AST,   we write \textsf{Edges}($\partialprog$) to denote the set of all edges in $\partialprog$,  \textsf{Root}($\partialprog$) to denote the root node, and \textsf{Subtree}($\partialprog, v$) to denote the subtree of $\partialprog$ rooted at node $v$.  Given a node $v$, we write $v:\nlabel$, to denote that the label of $v$ is $\nlabel$.
Adding a node $v:\nlabel$ to $\partialprog$ is denoted as $\partialprog[v \annot \nlabel]$ (in case $v$ already exists in $\partialprog$, it updates $v$'s label to be $\nlabel$). 
Furthermore, adding multiple nodes $v_1:\nlabel_1, \dots, v_n:\nlabel_n$  is denoted as $\partialprog[v_1 \annot \nlabel_1, \dots, v_n \annot \nlabel_n]$, and we assume that $(v_1, v_2), \dots, (v_1, v_n)$ are added as edges to $P$ if it does not already contain them.

\input{synthesis_alg}

 With this notation in place, we now explain the {\sc Synthesize} procedure from \figref{sketchcompletion}  in more detail. The algorithm first initializes the worklist to be the singleton $\{\partialprog_0\}$, where $\partialprog_0$ is a partial regex with a single node $v_0$ labeled with the input sketch $\sketch$ (line 2). The loop in lines 3--15 dequeues one of the partial regexes $\partialprog$ from the worklist and processes it based on whether it is concrete, symbolic, or neither. If it is concrete  (line 5), we return $\partialprog$ as a solution if it is consistent with the examples (line 6).
 
On the other hand, if $\partialprog$ is symbolic (line 7), we invoke a procedure called {\sc InferConstants} (described in Section~\ref{sec:inferConstants}) that instantiates the symbolic integers in $\partialprog$ with integer constants (line 8). 
As mentioned in Section~\ref{sec:overview}, {\sc InferConstants} should be viewed as merely a way of \emph{pruning} infeasible programs, so the  regexes produced by {\sc InferConstants} are \emph{not} guaranteed to  satisfy the examples. Thus, the regexes produced by {\sc InferConstants} still have to be checked for consistency with the examples in future iterations.
 
Lines 10-15 of the {\sc Synthesize} algorithm deal with the case where the dequeued partial regex is neither concrete nor symbolic (i.e., $\partialprog$ has at least one open node). In this case, we pick one of the open nodes  $v$ in $\partialprog$ and expand it according to the hints contained in the h-sketch labeling $v$. Specifically, the {\sc Expand} function from line 11 is described  in \figref{expandalg} using inference rules of the  form
$
 v : \sketch \vdash \partialprog \leadsto \exps 
 $.
 The meaning of this judgement is that we obtain a new set of partial regexes $\exps$ by expanding node $v$ according to  h-sketch $\sketch$ . Intuitively, given a node $v$ labeled with sketch $\hole_{d}\{\overline{\sketch}\}$, the inference rules enforce  that \emph{at least} one  descendant of $v$ must correspond to a regex in the languages of $\overline{\sketch}$.

 \input{expand_rules}

Next, given each expansion $\partialprog'$ of $\partialprog$, we check whether $\partialprog'$ is consistent with the provided examples via the call at line 13 to the {\sc Infeasible} function (discussed in  detail in Section~\ref{sec:over-under}). Observe that the worklist only contains  partial regexes that are consistent with the examples according to the abstract semantics given in Section~\ref{sec:over-under}.


\subsection{Pruning infeasible partial regexes}\label{sec:over-under}

The high-level idea for pruning infeasible partial regexes is quite simple and leverages the same observation made by~\citet{lee}:   Given a partial regex $\partialprog$, we can generate two concrete regexes, $\overapprox$ and $\underapprox$, that over- and under-approximate $\partialprog$ respectively. Specifically, $\overapprox$ and $\underapprox$ have the following properties:
\[
  \small
\begin{array}{ll}
(1) & \forall s. \ \  (\exists r \in \semantics{\partialprog}.\ \textsf{Match}(r, s)) \Rightarrow \textsf{Match}(\overapprox, s) \\
(2) & \forall s. \ \  \textsf{Match}(\underapprox, s) \Rightarrow (\forall r \in \semantics{\partialprog}.\ \textsf{Match}(r, s)) 
\end{array}
\]
Here, we use the notation $r \in \semantics{\partialprog}$ to denote that $r$ is a valid completion of $\partialprog$. Thus, $\overapprox$ matches every string $s$ that \emph{some} completion of $\partialprog$ can match and $\underapprox$ only matches those strings that \emph{all} completions of $\partialprog$ accept. Then, if there is any $\posexample \in \posexamples$ that $\overapprox$ does not match,  we know that $\partialprog$ cannot satisfy the examples and can be rejected without sacrificing completeness of our synthesis algorithm. Conversely, if there is any $\negexample \in \negexamples$ that $\underapprox$ matches, we know that $\partialprog$ will also match it and can thus be rejected safely.  {The main novelty of our feasibility checking technique compared to~\citet{lee} is to leverage the hints inside the h-sketch to compute more precise over- and under-approximations.}

\input{approxrules}

\figref{approxrules} describes our approximation procedure  using  inference rules of the  shape
$
\vdash \partialprog  \leadsto \langle \overapprox, \underapprox \rangle
$ indicating that $\partialprog$ is over- (resp. under-) approximated by $\overapprox$ (resp. $\underapprox$). These rules make use of an auxiliary judgment $\vdash S \twoheadrightarrow \langle \overapprox, \underapprox \rangle$ (described in \figref{approx-sketch}) that generate over- and under-approximations of hierarchical sketches. In what follows, we explain a subset of these rules.

\paragraph{Approximating holes.} 
The first three rules in ~\figref{approx-sketch} describe how to approximate holes in an h-sketch. We differentiate between two cases: 
If the depth of the hole is exactly $1$, then the hole must be filled with an instantiation of one of the h-sketches $\overline{\sketch}$. Thus, we first recursively compute over- and under-approximations for each $\sketch_i$ as $\langle o_i, u_i \rangle$. 
Then, the over-approximation for the hole is obtained by taking the union over all the $o_i$'s and the under-approximation is obtained by intersecting all the $u_i$'s (rule 3). The intuition for the latter is that the under-approximation must match only strings that \emph{every} instantiation of $\sketch_i$ matches; hence, we use intersection. 
On the other hand, for holes with depth greater than $1$,  we approximate them as $\langle \top, \bot \rangle$ (rule 2). 
In principle, we could perform a more precise approximation by instantiating the hole with \emph{every} possible DSL operator and taking the union/intersection of these regexes. However, since the resulting regex would be very large, such an alternative approximation would add a lot of overhead. 
Furthermore, since  holes can be nested inside one another,  we can often  obtain a useful approximation of the top-level sketch even when we use this less precise approximation for nested holes.



\paragraph{Approximating negation.} Rule 3 from \figref{approxrules} and rule 5 from \figref{approx-sketch} both deal with the negation operator. Because the negation of an over-approximation yields an under-approximation and vice versa,  {\tt Not}($\sketch$) is approximated as $\langle {\tt Not}(u), {\tt Not}(o) \rangle $ where $\langle o, u \rangle$ is the approximation for $\sketch$.

\paragraph{Approximating repetition operators.} The last two rules in ~\figref{approxrules} deal with operators in the {\tt Repeat} family, which take a regex as their first argument and integers for the remaining arguments. In rule 4, if all of the integer arguments are constants (rather than symbolic integers), then the over- and under-approximations are computed precisely.  However, if one of the arguments is a symbolic integer (rule 6),  the under-approximation is given by $\bot$, and the over-approximation is {\tt RepeatAtLeast}($o_1, {\tt 1}$) where $o_1$ is the over-approximation of the first argument. (Note that the second argument is $1$ since the integer arguments of all constructs in the {\tt Repeat} family require \emph{positive} integers.)

\input{approx_sketch}

\begin{example}
Consider the partial regex from Figure~\ref{fig:partial-expanded}. Its over-approximation is {\tt Concat(<num>, KleeneStar<any>)} and its under-approximation is shown in Eq.~\ref{eq:underapprox}.
\end{example}

\begin{theorem}{\bf (Correctness of \textnormal{\textsc{Approximate}} in \figref{approxrules})}
  Given a partial regex $\partialprog$, suppose \textnormal{\textsc{Approximate}}($\partialprog$) yields  $\langle  \overapprox, \underapprox  \rangle$. Then, we have:
  \begin{center}
  $\textnormal{(1)} \  \forall s. \   (\exists r \in \semantics{\partialprog}.\ \textsf{Match}(r, s)) \Rightarrow \textsf{Match}(\overapprox, s)$ \\
  $\textnormal{(2)} \  \forall s. \   \textsf{Match}(\underapprox, s) \Rightarrow (\forall r \in \semantics{\partialprog}.\ \textsf{Match}(r, s))$ 
  \end{center}
  \theoremlabel{approx1} 
\end{theorem}

\subsection{Solving Symbolic Regexes with SMT}\label{sec:inferConstants}

Recall that our method uses symbolic  regexes to avoid  explicit enumeration of integer constants that appear inside {\tt Repeat} constructs. In this section, we explain how to ``solve'' for these symbolic integers using SMT-based reasoning.

\figref{symint} shows the {\sc InferConstants} procedure for obtaining a set of concrete regexes from a given symbolic regex $\partialprog$. The high-level idea underlying this algorithm is as follows: We first infer a constraint $\phi$ on the values of symbolic integers $\kappa_1, \ldots, \kappa_n$ using the \emph{length} of the strings that appear in the examples. However, this constraint is over-approximate 
in the sense that every concrete regex must satisfy $\phi$ but not every model of $\phi$ corresponds to a concrete regex that satisfies the examples. Thus, given a candidate assignment to one of the $\kappa$'s (obtained from a model of $\phi$), we use the {\sc Infeasible} procedure discussed in the previous section to check whether this (partial) assignment is feasible. If so, we then continue and repeat the same process for the remaining $\kappa_i$'s until we have found a full assignment for all symbolic integers that appear in $\partialprog$. 

\input{inferconstants}

\paragraph{SMT Encoding}
Before explaining the {\sc InferConstants} algorithm in more detail, we first explain how to generate a constraint for a given symbolic regex.
Our encoding is described  in \figref{constructrules} using inference rules $\partialprog \hookrightarrow (\phi, x)$. The meaning of this judgment is that, for any instantiation of $\partialprog$ to match a string $s$, the symbolic integers occurring in $\partialprog$ must satisfy $\phi[\emph{len}(s)/x]$. As is evident from the first rule in \figref{constructrules}, our encoding makes use of a function $\Phi$, shown also in  \figref{constructrules}, that generates a constraint for a given regex from constraints on its sub-regexes. Specifically, it takes as input a DSL construct {\tt op}, a variable $x$ that refers to the length of the string matched by the top-level regex, and constraints $\phi_1, \ldots, \phi_k$ for the sub-regexes (where the length of the string matched by $i$'th sub-regex is $x_i$). 

For instance, consider the encoding for the {\tt StartsWith}($r$) construct: If the length of the string matched by $r$ is $x_1$ (which is constrained according to $\phi_1$), then any string matched by {\tt StartsWith(r)} will be at least as long as $x_1$. Thus, we have:
\[
\Phi({\tt StartsWith}, x, x_1, \phi_1) = \exists x_1. (x \geq x_1 \land \phi_1)
\]
Observe that $x_1$ is existentially quantified in the formula because it is a ``temporary'' variable that refers to the length of the string matched by the sub-regex. Since the other cases in the definition of the $\Phi$ function are similar and follow the semantics of the DSL operators, we do not discuss them in detail but just highlight two cases for {\tt Not} and {\tt RepeatAtLeast}. 

The encoding for the {\tt Not} operator is \emph{true} regardless of the sub-regex because inferring anything more precise would require us to track \emph{sufficient} (rather than necessary) conditions for accepting a string, which is not feasible to do using the \emph{length} of the string alone.

The encoding for the {\tt Repeat} family of constructs introduces non-linear multiplication. 
For instance, consider  the symbolic regex {\tt RepeatAtLeast}($r,\kappa$) where the constraint on the sub-regex $r$ is $(\phi_1,x_1)$. Since $r$ is repeated at least $\kappa$ times, the length of the string matched by this regex is at least  $x_1 \cdot \kappa$, which introduces non-linear constraints. 
Thus, while the formulas generated by the {\sc Encode} procedure are technically in Peano (rather than Presburger) arithmetic, we found that the Z3 SMT solver  can efficiently handle the type of non-linear constraints we generate.

\begin{example}
  Consider the following symbolic regex:
  \small
  \begin{align}\label{eq:ex-sym}
    &{\tt Concat} \Big( {\tt Repeat} \big( {\tt Or}({\tt <.>},{\tt <num>}),\kappa_1 \big), \\[-0.8em]
    &{\tt RepeatAtLeast} \big( {\tt RepeatRange}({\tt <num>},1,3),\kappa_2 \big) \Big)\nonumber
  \end{align}

  \normalsize
  Using the rules presented in \figref{constructrules}, we  generate the following constraint $\phi$:
  \small
  \begin{align}
    \tag{\texttt{Concat}} \phi   &= \exists x_1, x_2. \ (x_0 = x_1 + x_2) \wedge \phi_1 \wedge \phi_2 \\ 
    \tag{\texttt{Repeat}} \phi_1 &= \exists x_3,x_3'. \ (x_1 \geq x_3 * \kappa_1 \wedge x_1 \leq x_3' * \kappa_1) \\ & \hspace{35pt} \wedge \phi_3 \wedge \phi_3[x_3'/x_3]\wedge (1 \leq \kappa_1 \leq \textsf{MAX}) \nonumber \\
    \tag{\texttt{Or}} \phi_3 &= (x_3 = 1 \vee x_3 = 1) \\
    \tag{\texttt{AtLeast}} \phi_2 &= \exists x_4. \ (x_2 \geq x_4 * \kappa_2) \wedge \phi_4 \wedge (1 \leq \kappa_2 \leq \textsf{MAX}) \\
    \tag{\texttt{Range}} \phi_4 &= 1 \leq x_4 \leq 3
  \end{align}
  \normalsize
 Note that the top-level constraint $\phi$ can be simplified to the following formula   by performing quantifier elimination: 
 \small
 \begin{equation}\label{eq:constraint-simp}
 (x_0 \geq \kappa_1 + \kappa_2) \wedge (1 \leq \kappa_1 \leq \textsf{MAX}) \wedge (1 \leq \kappa_2 \leq \textsf{MAX})
 \end{equation}
 \normalsize
\end{example}

\paragraph{Using SMT encoding for inference}

Now that we have a way to encode symbolic regexes using  SMT, we can describe the {\sc InferConstants} algorithm from  \figref{symint} in more detail. Given a symbolic regex $\partialprog_0$, the algorithm first generates the SMT encoding $\phi_0$ for $\partialprog_0$ using the {\sc Encode} function  (i.e., \figref{constructrules}). Here, $\phi_0$ contains free variables $\kappa_1, \ldots, \kappa_n$ as well as a variable $x_0$ that refers to the length of the input string. Now, since every $s \in \posexamples$ should match the synthesized regex, we can obtain a constraint on the symbolic integers by instantiating $x_0$ with $\emph{len(s)}$ for every $s \in \posexamples$ and taking their conjunction.  Thus, formula $\psi_0$ from line 2 gives us a constraint on the symbolic integers used in~$\partialprog$.

\input{encode_rules}



Next, the loop in lines 5--13 populates a set $\Pi$ of concrete regexes  that can be obtained by instantiating the symbolic integers in $\partialprog_0$ with constants. Towards this goal, it maintains a worklist of symbolic regexes that are made increasingly more concrete in each iteration.

Specifically, the worklist contains pairs $(\partialprog, \phi)$ where $\partialprog$ is a symbolic regex and $\phi$ is a constraint on the symbolic integers used in $\partialprog$ --- initially, the worklist just contains $(\partialprog_0, \psi_0)$. Then, in each iteration, we remove from the worklist a symbolic regex $\partialprog$  and its constraint $\phi$ and make an assignment to one of the symbolic integers $\kappa$ used in $\partialprog$. To this end, we first query the SMT solver to get a model $\sigma$ of $\phi$. However, since $\phi$ is over-approximate, instantiating the symbolic integers in $\partialprog$ with $\sigma$ may not yield a concrete regex that satisfies the examples. Thus, we pick one of the symbolic integers $\kappa$ in $\partialprog$ and check whether $\sigma[\kappa]$ is infeasible using the method described in Section~\ref{sec:over-under} (line 12).~\footnote{Alternatively, we could plug in the whole assignment $\sigma$ and 
check whether the resulting regex is consistent with the examples. However, our proposed method is preferable over this alternative because a partial assignment to a subset of the variables often results in an infeasible partial regex and allows us to prune significantly more programs.
} If the resulting symbolic regex cannot be proven infeasible, we then add the partially concretized symbolic program $P' = P [\kappa \annot \sigma[\kappa]] $ to the worklist, together with its corresponding constraint $\phi[\kappa \annot \sigma[\kappa]]$ (line 13
). However, in addition, we also keep the original symbolic regex $\partialprog$ since there may be other valid assignments to $\kappa$ beyond just $\sigma[\kappa]$ (line 9). Finally, to ensure that the solver does not keep yielding the same assignment to $\kappa$, we strengthen its constraint by adding the ``blocking clause'' $\kappa \neq \sigma[\kappa]$ (also line 9). Upon termination, the set $\Pi$ contains every feasible concrete regex that can be obtained by instantiating the original symbolic regex $\partialprog_0$.

\begin{example}
 Consider the simplified constraint $\phi$ from Eq.~\ref{eq:constraint-simp}. After instantiating $x_0$  with the length of each positive example from Section~\ref{sec:overview} and taking their conjunction, we obtain the following formula $\psi_0$:
 \small 
 \begin{align*}
    &(\kappa_1 + \kappa_2  \leq 13) \wedge (\kappa_1 + \kappa_2 \leq 7) \wedge (\kappa_1 + \kappa_2 \leq 18) \wedge (\kappa_1 + \kappa_2 \leq 15) \\& \wedge (1 \leq \kappa_1 \leq \textsf{MAX}) \wedge (1 \leq \kappa_2 \leq \textsf{MAX})
  \end{align*}
\normalsize      
This formula  is equivalent to the following much simpler constraint:
\small
\begin{equation}\label{eq:simp-inst}
\psi_0 = (\kappa_1 + \kappa_2 \leq 7) \wedge (1 \leq \kappa_1 \leq \textsf{MAX}) \wedge (1 \leq \kappa_2 \leq \textsf{MAX} )
\end{equation}
\normalsize
 Now, suppose the solver returns the model $[\kappa_1 \mapsto 1, \kappa_2 \mapsto 1]$ to Eq.~\ref{eq:simp-inst}. Thus, we first assign $1$ to $\kappa_1$ in the partial regex from Eq.~\ref{eq:ex-sym}, which yields: 
 \small
 \begin{align*}
  &{\tt Concat} \Big( {\tt Repeat} \big( {\tt Or}({\tt <num>},{\tt <.>}),1 \big), \\[-0.8em]
  &{\tt RepeatAtLeast}({\tt RepeatRange} \big( {\tt <num>},1,3),\kappa_2 \big) \Big)
 \end{align*}
\normalsize
We can prove that this partial regex is inconsistent with the examples from Section~\ref{sec:overview} because no instantiation of $\kappa_2$ yields a regex that matches the positive example ``123456789.123''. Observe that ignoring the assignment to $\kappa_2$ allows us to prune $6$ regexes at a time instead of just one.
\end{example}

\begin{theorem}{\bf (Correctness of \textnormal{\textsc{InferConstants}} in \figref{symint})}
  Given a partial regex $\partialprog$, positive examples $\posexamples$ and negative examples $\negexamples$, suppose that \textnormal{\textsc{InferConstants}} returns $\exps$. 
  Then, for any concrete regex $\prog \in \semantics{\partialprog}$ that is consistent with $\posexamples$ and $\negexamples$, we have $\prog \in \exps$. 
  \end{theorem}

%% file: synthesis_alg.tex
\begin{figure}
\small
\begin{algorithm}[H]
\begin{algorithmic}[1]
\Procedure{$\completesketch$}{$\sketch, \posexamples, \negexamples$}
\Statex \Input{an h-sketch $\sketch$, positive and negative examples $\posexamples$, $\negexamples$}
\Statex \Output{a regex consistent with $\sketch$, $\posexamples$ and $\negexamples$, or $\bot$}
\vspace{0.05in}
\State $\partialprog_0 \assign ( v_0, \emptyset, [v_0 \annot \sketch] )$; \ $\worklist \assign \{ \partialprog_0 \}$; 
\vspace{0.05in}
\While{$\worklist \neq \emptyset$}
\State $\partialprog \assign \worklist.\textsf{remove()}$; 
\If{\textsf{IsConcrete}($\partialprog$)}
\If{\textsf{IsCorrect}($\partialprog, \posexamples, \negexamples$)}
\Return $\partialprog$;  
\EndIf
\ElsIf{\textsf{IsSymbolic($\partialprog$)}} 
\State $\worklist \assign \worklist \cup \textsc{InferConstants}(\partialprog, \posexamples, \negexamples)$; 
\Else 
\State $(v, \sketch) \assign \textsf{SelectOpenNode}(\partialprog)$;
\State $\worklist' \assign \textsc{Expand}(\partialprog, v, \sketch)$;
\ForAll{$\partialprog' \in \worklist'$} 
\If{\textsc{Infeasible($\partialprog', \posexamples, \negexamples$)}} 
\State $\worklist'.\textsf{remove}(\partialprog')$;
\EndIf
\EndFor
\State $\worklist \assign \worklist \cup \worklist'$;
\EndIf 
\EndWhile 
\State \Return $\bot$;
\EndProcedure
\end{algorithmic}
\end{algorithm}
\vspace{-15pt}
\caption{Synthesis algorithm for generating a regex from an h-sketch and a set of positive/negative examples.}
\figlabel{sketchcompletion} 
\vspace{-6pt}
\end{figure}

%% file: expand_rules.tex
\begin{figure}
\small 
\[
\hspace{-15pt}
\begin{array}{cr}
\begin{array}{cc}
\irule{
\begin{array}{c}
n = |\overline{\sketch}|  \quad \quad
\exps =  \bigcup_{i=1}^{n} \big\{  \partialprog[v \annot \sketch_i]   \big\} 
\end{array}
}{
v : \hole_{1}\{ \overline{\sketch}\} \vdash \partialprog \leadsto \exps 
} \ \ {\rm (1)}
\end{array} 
\\ \\
\begin{array}{cc}
  \irule{
\begin{array}{c}
  \exps_1 = \bigcup_{i=1}^{|\overline{\sketch}|} \big\{  \partialprog[v \annot \sketch_i]  \big\} \ \ \ \ \    \nlabel = \hole_{d-1}\{ \overline{\sketch} \} \ \ \ \ \ \nlabel' = \hole_{d-1} \mathcal{C} \cup \{ \overline{\sketch} \}  
\\ 
\exps_2 = \bigcup_{j=1}^{|\overline{v}|} \big\{ \partialprog[v \annot \texttt{f},   v_j \annot \nlabel, \forall_{i \neq j } v_i \annot \nlabel']  \ | \ \texttt{f} \in \mathcal{F}_{|\overline{v}|}, \overline{v}  \text{ fresh} \big\}
\\ 
\exps_3 = \big\{ \partialprog[v \annot \texttt{g},  v_0 \annot \nlabel, \forall_{i \in [1, |\overline{v}|]} v_i \annot \kappa_i]  \ | \ \texttt{g} \in \mathcal{G}_{|\overline{v}|}, \overline{v}, \overline{\kappa} \text{ fresh}  \big\}
\end{array}
}{
v : \hole_{d\geq 1}\{ \overline{\sketch} \} \vdash \partialprog \leadsto \exps_1 \cup \exps_2 \cup \exps_3 
} \ \ {\rm (2)}
\end{array}
\\ \\ 
\begin{array}{cc}
\begin{array}{cc}
\irule{
\begin{array}{c}
\overline{v} \text{ fresh} \quad \quad  n = |\overline{\sketch}| \quad \quad
\exps = \big\{\partialprog[v \annot \texttt{f}, \forall_{i \in [1, n]}. v_i \annot \sketch_i] \big\}
\end{array}
}{
v : \texttt{f} (\overline{\sketch}) \vdash \partialprog \leadsto \exps
} \ \ {\rm (3)}
\end{array}
\\ \\
\begin{array}{cc}
\irule{
\begin{array}{c}
\overline{v} \text{ fresh}\ \quad \quad
\exps = \big\{ \partialprog[v \annot \texttt{g}, v_0 \annot \sketch, \forall_{i \in [1, |\overline{\kappa}|]}. v_i \annot \kappa_i]   \big\} 
\end{array}
}{
v : \texttt{g} (\sketch, \overline{\kappa})  \vdash \partialprog \leadsto  \exps 
} \ \ {\rm (4)}
\end{array}
\end{array}
\end{array}
\]
\caption{Inference rules for \textsc{Expand}. In rule (2), $\mathcal{C}$ denotes all character classes in the DSL, $\mathcal{G}_i$  (resp. $\mathcal{F}_i$)  denotes \texttt{Repeat} (resp. non-\texttt{Repeat}) constructs  with arity $i$. }
\figlabel{expandalg}
\end{figure}

%% file: approxrules.tex
\begin{figure}
\small
\[
\hspace{-8pt}
\begin{array}{cr}
\begin{array}{c}
\irule{
\begin{array}{c}
Root(\partialprog) = v:\sketch \ \ \ \  \vdash \sketch \twoheadrightarrow \langle \overapprox, \underapprox \rangle
\end{array}
}{
\vdash \partialprog    \leadsto \langle \overapprox, \underapprox \rangle
} \ \ {\rm \text{(1)}} 
\end{array} 
\\ \\ 
\begin{array}{c}
\irule{
\begin{array}{c}
Root(\partialprog) = v:(\texttt{f} \in \mathcal{F}_n \setminus \{\texttt{Not}\})\\ (v, v_i) \in \textsf{Edges}(\partialprog)   \ \ \ \ \ \    \vdash \textsf{Subtree}(\partialprog, v_i)   \leadsto \langle  \overapprox_i, \underapprox_i   \rangle    \\ 
\end{array}
}{
 \vdash \partialprog  \leadsto \big\langle  \texttt{f} (\overline{\overapprox_N}), \texttt{f} (\overline{\underapprox_N}) \big\rangle
} \ \ {\rm {\text{(2)}}}
\end{array}
\\ \\ 
\begin{array}{c}
  \irule{
  \begin{array}{c}
  Root(\partialprog) = v:\texttt{Not}  \\  (v, v_1) \in \textsf{Edges}(\partialprog)   \ \ \ \ \ \    \vdash \textsf{Subtree}(\partialprog, v_1)   \leadsto \langle  \overapprox_1, \underapprox_1   \rangle    \\ 
  \end{array}
  }{
   \vdash \partialprog  \leadsto \big\langle  \texttt{Not} (\underapprox_1), \texttt{Not} (\overapprox_1) \big\rangle
  } \ \ {\rm {\text{(3)}}}
  \end{array}
  \\ \\ 
\begin{array}{c}
\irule{
\begin{array}{c}
Root(\partialprog) = v:(\texttt{g} \in \mathcal{G}_{n})   \quad    (v, v_i:\nlabel_i) \in \textsf{Edges}(\partialprog)   \\    \vdash \textsf{Subtree}(\partialprog, v_1) \leadsto \langle  \overapprox_1, \underapprox_1 \rangle   \ \ \ \ \  \forall i \geq 2. \ \nlabel_i \in  \mathbb{N}
\end{array}
}{
 \vdash \partialprog  \leadsto  \big\langle \texttt{g}(\overapprox_1, \overline{\nlabel} ),  \texttt{g}(\underapprox_1, \overline{\nlabel} )    \big\rangle
} \ \ {\rm {\text{(4)}}}
\end{array}
\\ \\ 
\begin{array}{c}
\irule{
\begin{array}{c} 
Root(\partialprog) = v:(\texttt{g} \in \mathcal{G}_{n})  \quad     (v, v_i:\nlabel_i) \in \textsf{Edges}(\partialprog)  \\      
\vdash \textsf{Subtree}(\partialprog, v_1) \leadsto \langle  \overapprox_1, \underapprox_1 \rangle   \ \ \ \ \     \exists i \geq 2. \ \textsf{SymInt}(\nlabel_i) 
\end{array}
}{
\vdash \partialprog    \leadsto  \big\langle \texttt{RepeatAtLeast}(o_1, 1), \bot       \big\rangle
} \ \ {\rm {\text{(5)}}}
\end{array}
\end{array}
\]
\vspace{-5pt}
\caption{Inference rules for \textsc{Approximate}.  $\mathcal{G}_n$ (resp. $\mathcal{F}_n$)  denotes arity $n$ operators in (resp. not in) the \texttt{Repeat} family. }
\figlabel{approxrules} 
\end{figure}

%% file: approx_sketch.tex
\begin{figure}[t]
 \small
  \[
  \hspace{-10pt}
  \begin{array}{cr}
  \begin{array}{c}
  \irule{
  \begin{array}{c}
  \vdash \sketch \twoheadrightarrow \langle\overapprox, \underapprox \rangle 
  \end{array}
  }{
  \vdash \hole_1\{\sketch\}  \twoheadrightarrow \langle\overapprox, \underapprox \rangle 
  } \ \ {\rm \text{(1)}}
  \ \ \ \ \ \ \ \ \
  \irule{
  \begin{array}{c}
  d > 1
  \end{array}
  }{
  \vdash \hole_d\{\overline{\sketch}\}  \twoheadrightarrow \langle \top, \bot \rangle 
  } \ \ {\rm \text{(2)}}
  \end{array}
  \\ \\
  \begin{array}{c}
    \irule{
      \begin{array}{c}
       \vdash \sketch_1 \twoheadrightarrow \langle \overapprox, \underapprox \rangle 
      \ \ \ \ \  \vdash \hole_1 \{ \sketch_2, \ldots, \sketch_{|\overline{\sketch}|} \} \twoheadrightarrow \langle o', u' \rangle
      \end{array}
      }{
      \vdash \hole_1\{\overline{\sketch}\}  \twoheadrightarrow \langle \texttt{Or}(\overapprox, \overapprox'), \texttt{And}(\underapprox, \underapprox') \rangle 
      } \ \ {\rm \text{(3)}}
    \end{array}
  \\ \\ 
  
  \begin{array}{c}
  \irule{
  \begin{array}{c}
  \texttt{f} \in \mathcal{F}_n \setminus \{\texttt{Not}\} \ \ \ \ \ \vdash \sketch_i \twoheadrightarrow \langle \overapprox_i, \underapprox_i \rangle
  \end{array}
  }{
  \vdash \texttt{f}(\overline{\sketch}) \twoheadrightarrow \langle \texttt{f}(\overline{\overapprox}), \texttt{f}(\overline{\underapprox}) \rangle 
  } \ \ {\rm \text{(4)}} \\ \\
  
  \irule{
  \begin{array}{c}
   \vdash \sketch \twoheadrightarrow \langle \overapprox, \underapprox \rangle
  \end{array}
  }{
  \vdash \texttt{Not}(\sketch) \twoheadrightarrow \langle  \texttt{Not}(\underapprox), \texttt{Not}(\overapprox) \rangle 
  } \ \ {\rm \text{(5)}}
  
    \ \ \ \ \ \
  \irule{
  \begin{array}{c}
  
  \end{array}
  }{
  \vdash r  \twoheadrightarrow \langle r, r \rangle 
  } \ \ {\rm \text{(6)}}
  
  \\ \\
  
  \irule{
  \begin{array}{c}
    \texttt{g} \in \mathcal{G}_{n}  \ \ \ \ \  \vdash \sketch \twoheadrightarrow \langle \overapprox, \underapprox \rangle
  \end{array}
  }{
  \vdash \texttt{g}(\sketch,\overline{\kappa})  \twoheadrightarrow \langle  \texttt{RepeatAtLeast}(\overapprox,1), \bot \rangle 
  } \ \ {\rm \text{(7)}}

  \end{array}
  \end{array}
    \]
  \vspace{-5pt}
  \caption{Inference rules for over- and under-approximating h-sketches. $\prog$ denotes a concrete regex.}.
  \figlabel{approx-sketch} 
  \end{figure}

%% file: inferconstants.tex
\begin{figure}\small
\begin{algorithm}[H]
\begin{algorithmic}[1]
\Procedure{InferConstants}{$\partialprog_0, \posexamples, \negexamples$}
\Statex \Input{a symbolic regex $\partialprog_0$, examples $\posexamples$, $\negexamples$.}
\Statex \Output{a set of concrete regular expressions $\exps$.}
\vspace{0.05in}
\State $(\phi_0, x_0) \assign \textsf{Encode}(\partialprog_0)$; \ \ $\psi_0 \assign \big( \bigwedge_{s \in \posexamples} \phi_0[\textsf{len}(s)/x_0]   \big)  $;
\State $\worklist \assign \{ ( \partialprog_0, \psi_0 ) \}$;  \ \ \ \ \   $\exps \assign \emptyset$; 
\While{$\worklist \neq \emptyset$} 
\State $( \partialprog, \phi ) \assign \worklist.\textsf{remove}()$;  
\If{$\textsf{UNSAT}(\phi)$} 
\textbf{continue}; 
\EndIf 
\State $\sigma \assign \textsf{Model}(\phi)$; $\kappa \assign \textsf{ChooseSymInt}(\partialprog)$;
\State $\partialprog' \assign \partialprog \big[ \kappa \annot \sigma[\kappa] \big]$;
\State $\worklist \assign \worklist \cup \{ ( \partialprog,   \phi \wedge \kappa  \neq \sigma [\kappa] )  \}$;
\If{$\textsf{IsConcrete}(\partialprog')$} 
 $\exps \assign \exps \cup \{  \partialprog'  \}$; 
\Else 
\If{$\neg \textsc{Infeasible}(\partialprog, \posexamples, \negexamples)$} 
\State $\worklist \assign \worklist \cup \{  (\partialprog'  ,   \phi \big[  \kappa \annot \sigma[\kappa] \big] )   \}$;
\EndIf 
\EndIf
\EndWhile 
\State \Return $\exps$;
\EndProcedure
\end{algorithmic}
\end{algorithm}
\vspace{-10pt}
\caption{Algorithm for \textsc{InferConstants}.}
\figlabel{symint} 
\vspace{-5pt}
\end{figure}

%% file: encode_rules.tex
\begin{figure}
\small 
\[
\hspace{-25pt}
\begin{array}{cr}
\begin{array}{c}
\irule{
\begin{array}{c}
\emph{Root}(\partialprog) = v : \texttt{op}     \ \ \ \ \ \ \ 
\textsf{arity}(\texttt{op}) = n   \ \ \ \ \ \ \  
(v, v_i) \in \textsf{Edges}(\partialprog) \\
\textsf{Subtree}(\partialprog,v_i)   \hookrightarrow  (\phi_i, x_i)  \ \ \ \ \ \ \ 
  x \ \emph{fresh}
\end{array}
}{
 \partialprog  \hookrightarrow  (\Phi(\texttt{op}, x, \overline{x}, \overline{\phi}), \ x)
} \ \ {\rm {\text{(1)}}}   
\\ \\ 
\ \ \ \ \ \ 
\irule{
\begin{array}{c}
x \ \emph{fresh} \ \ \ 
\emph{Root}(\partialprog) = v : ( c \in \mathcal{C} )  \\
\end{array}
}{
 \partialprog  \hookrightarrow  (x = 1, \ x)
} \ \ {\rm {\text{(2)}}}  
\\ \\ 
\irule{
\begin{array}{c}
\emph{Root}(\partialprog) = v : \big( \kappa \in \textsf{SymInt}(\partialprog) \big) 
\end{array}
}{
 \partialprog   \hookrightarrow  (1 \leq \kappa \leq \textsf{MAX}, \ \kappa)
} \ \ {\rm {\text{(3)\footnotemark}}}  
\end{array}
\end{array}
\]
\\ 
\[
\hspace{-10pt}
\begin{array}{lll}
\Phi(\texttt{StartsWith},  x, \overline{x}, \overline{\phi}) & = & \exists x_1. \ (x \geq x_1 \wedge \phi_1) \\ 
\Phi(\texttt{EndsWith}, x, \overline{x}, \overline{\phi}) & =  & \exists x_1.  \ (x \geq x_1 \wedge \phi_1) \\ 
\Phi(\texttt{Contains}, x, \overline{x}, \overline{\phi}) & = & \exists x_1. \ (x \geq x_1 \wedge \phi_1) \\  
\Phi(\texttt{Not}, x, \overline{x}, \overline{\phi}) & = & \emph{true} \\ 
\Phi(\texttt{Optional}, x, \overline{x}, \overline{\phi}) & = & \exists x_1. \ (x = 0 \vee x = x_1 ) \wedge \phi_1   \\ 
\Phi(\texttt{KleeneStar}, x, \overline{x}, \overline{\phi}) & = & \exists x_1. \ ( x = 0 \vee  x \geq x_1 ) \wedge \phi_1  \\ 
\Phi(\texttt{Concat}, x, \overline{x}, \overline{\phi}) & = &  \exists x_1, x_2. (x = x_1 + x_2) \\ & & \hspace{30pt} \wedge \phi_1  \wedge \phi_2  \\ 
\Phi(\texttt{Or}, x, \overline{x}, \overline{\phi}) & = &  \exists x_1, x_2.  ( x = x_1 \vee x = x_2) \\ & & \hspace{30pt} \wedge \phi_1 \wedge \phi_2  \\ 
\Phi(\texttt{And}, x, \overline{x}, \overline{\phi}) & = &   \exists x_1, x_2.  (x = x_1 \wedge x = x_2)  \\ & & \hspace{30pt} \wedge \phi_1 \wedge \phi_2  \\ 
\Phi(\texttt{Repeat}, x, \overline{x}, \overline{\phi}) & = &    \exists x_1, x_1'. (x \geq x_1 x_2 \wedge x \leq x_1' x_2)  \\ & & \hspace{30pt} \wedge \phi_1 \wedge \phi_1[x_1'/x_1] \wedge \phi_2  \\ 
\Phi(\texttt{RepeatAtLeast}, x, \overline{x}, \overline{\phi}) & = &    \exists x_1. (x  \geq x_1  x_2) \wedge \phi_1 \wedge \phi_2  \\ 
\Phi(\texttt{RepeatRange}, x, \overline{x}, \overline{\phi}) & = &    \exists x_1, x_1'. (x \geq x_1  x_2  \wedge   x   \leq  x_1'  x_3) \\ & & \hspace{10pt}   \wedge \phi_1 \wedge \phi_1[x_1'/x_1] \wedge \phi_2   \wedge  \phi_3  \\ 

\end{array}
\]

\vspace{-5pt}
\caption{Inference rules for \protect$\textsc{Encode}$.}
\figlabel{constructrules} 
\end{figure}
\footnotetext{ $\textsf{MAX}$ is the maximum integer constant in the DSL. We set $\textsf{MAX}$ to the length of the longest example in the implementation.}

%% file: sketch_generation.tex
\begin{figure*}
    \begin{center}
    \includegraphics[trim=8 50 8 20,clip,scale=0.5, width=\textwidth]{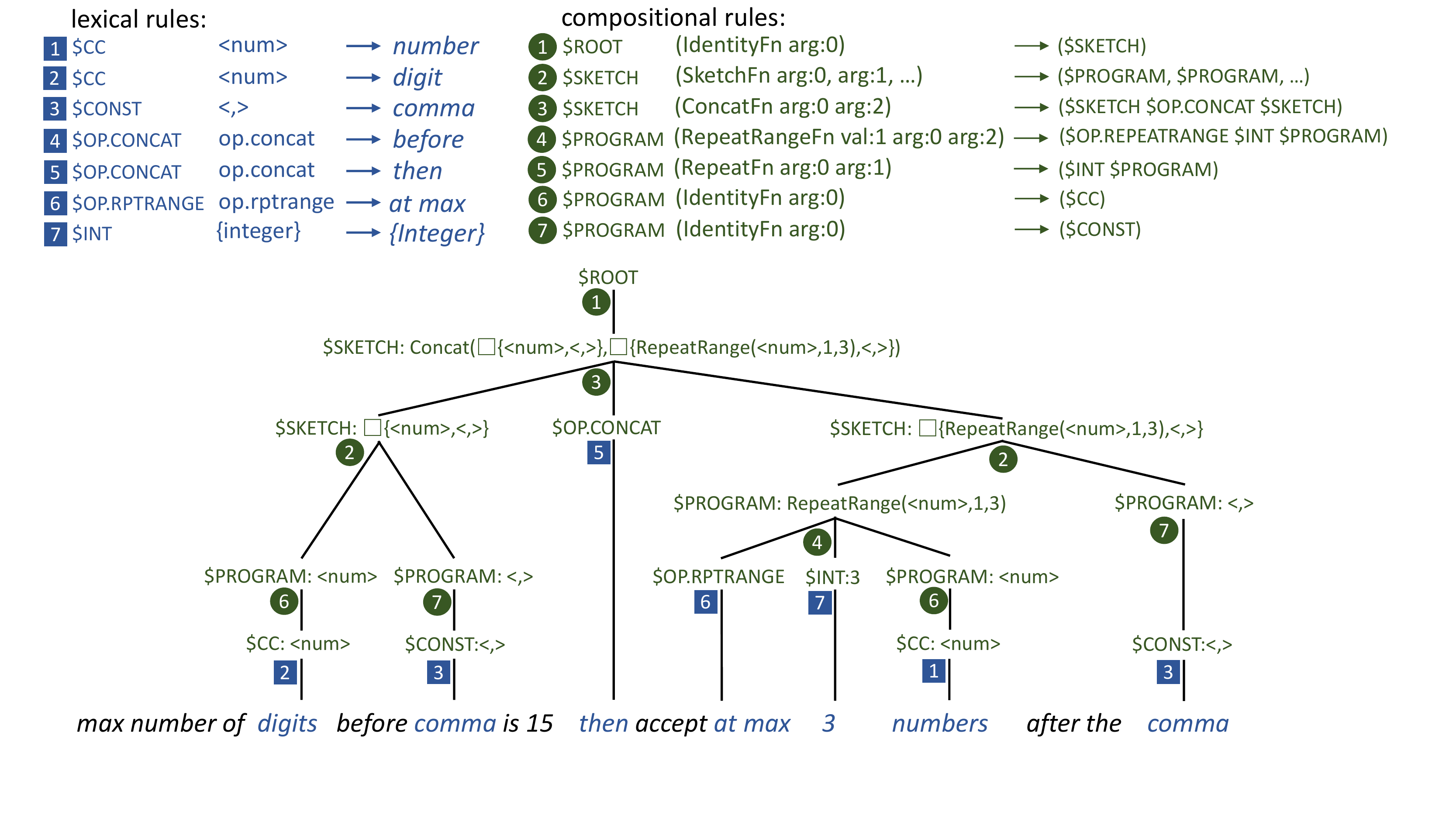}
    \end{center}
    \caption{Examples of rules and the parse tree for one possible derivation generated from the given description. 
    }
    \figlabel{semprerules}
\end{figure*}

\section{From English Text to H-Sketches}
\seclabel{sketchgen} 
In this section, we describe a technique for generating hierarchical sketches from English text. While there are many NLP techniques  that can be used to solve this problem (including  currently-popular \emph{seq2seq} models), we frame it as an instance of semantic parsing and build our sketch generator on top of the  SEMPRE framework~\cite{berant}.  As mentioned briefly in Section ~\ref{sec:intro}, we choose  semantic parsing over  deep learning techniques because it does not require as much labeled training data. However, our general synthesis methodology and the PBE algorithm are both agnostic to the NLP technique used for parsing English text into an h-sketch.

\subsection{Background on semantic parsing} 
Semantic parsing  is used for converting natural language to a formal representation, such as SQL~\cite{mooney,sqlizer}, lambda calculus~\cite{carpenter}, or natural logic~\cite{manning}. This formal representation is often referred to as a \emph{logical form}, and semantic parsers use a context-free grammar (CFG) to translate  natural language  to the target logical form. However, since natural language is highly complex and often very ambiguous, there are many possible logical forms that can be obtained from a given natural language description. Thus, modern semantic parsers also incorporate a machine learning model to score different parses for a given utterance. However, as mentioned earlier, these techniques still do not require as much labeled training data as other methods based on deep learning. 

In the context of this work, logical forms correspond to hierarchical sketches, so our CFG needs to parse a given English utterance into an h-sketch. In the remainder of this section, we first give an overview of \toolname's CFG (Section~\ref{sec:semparser-grammar}) and then  discuss how to produce a \emph{ranked} list of h-sketches using a machine learning model (Section~\ref{sec:semparser-learning}).

\subsection{Grammar-based sketch composition}\label{sec:semparser-grammar}


Following standard convention,  we specify our grammar rules in the following format:
\par
\emph{
        <target category> <target derivation> $\rightarrow$ <source sequence>
}\\
Such a rule maps  \emph{<source sequence>} to a \emph{<target derivation>} with category \emph{<target category>}. 
Rules of the semantic parser can be further categorized into two groups, namely lexical rules and compositional rules.
Examples of both types of rules are provided in \figref{semprerules}.
A lexical rule maps a word in the sentence to base concepts in the DSL, including character class (e.g., lexical rule 1) and operator (e.g., lexical rule 4).
A compositional rule combines one or more base components and builds larger h-sketches. For instance, as shown in \figref{semprerules}, compositional rule 2 is applied to generate a sketch $\hole{\{\texttt{<num>,<,>}\}}$, labeled with category \textit{\$SKETCH}, from a sequence of two derivations, $\texttt{<num>}$ and $\texttt{<,>}$, both labeled with \textit{\$PROGRAM}, via the semantic function \emph{SketchFn}. Here, we use category \textit{\$SKETCH} to denote sketches containing holes and category \textit{\$PROGRAM} to mark concrete regexes. 


Given  a set of pre-defined grammar rules and a natural language description $\mathcal{L}$, the semantic parser generates a  list of possible derivations for $\mathcal{L}$. Each derivation  can be mapped to an h-sketch deterministically, and, in general, multiple derivations of the same sentence can map to the same h-sketch.
We construct the derivations for a given sentence  recursively  in a bottom-up fashion using dynamic programming.
More specifically, 
we first apply lexical rules to generate derivations for any span (i.e., sequence of words) that they match. Then, the derivations of larger spans are constructed by applying compositional rules to derivations built over non-overlapping constituent spans. As the final output, we take derivations spanning the whole sentence that are labeled with a designated \textit{\$ROOT} category.

\begin{example}
To build intuition, \figref{semprerules} demonstrates the parsing process for the English phrase
\emph{``the max number of digits before comma is 15 then accept at max 3 numbers''}.
{Note that our parser allows skipping arbitrary words; thus, not every span in the description is used for building this derivation. Finally, we do not require applying every rule from ~\figref{semprerules} when constructing this derivation, such as lexical  rule 4 and compositional rule 5. }
Also observe that our grammar does not uniquely define an h-sketch for a given  sentence. In particular, we can also obtain 
the following alternative h-sketch from the same text:
\begin{equation}\label{eq:badsketch}
    \texttt{Concat} \big( \hole{\{\texttt{<num>}\}},\hole\{\texttt{<,>},\texttt{Repeat}(\texttt{<num>,3})\} \big)
    \end{equation}
\end{example}

\subsection{Learning feature weights}\label{sec:semparser-learning}

Since there are many different h-sketches for an given English sentence, we need a way of scoring derivations so that  h-sketches that are more consistent with the utterance are assigned a higher score. Towards this goal, our parser leverages a discriminative log-linear model using a set of \emph{features} extracted from natural language.
Specifically, given a derivation $d$ from the set of possible derivations $D(\mathcal{L})$ for a description $\mathcal{L}$, we extract a feature vector $\phi(\mathcal{L},d)\in \mathbb{R}^b$. 
The  features are local to individual rules and are chosen to capture lexical, compositional, and semantic characteristics of the derivation and its sub-derivations.
\toolname\ leverages two feature sets, namely \emph{rule features} and \emph{span features}, both of which are inherited from the SEMPRE framework.
Concretely, a rule feature indicates whether a particular rule is fired during the derivation, and a span feature tracks the number of consecutive words that are used when generating a particular category in the derivation.



Given these extracted feature vectors, the probability that a derivation $d$ is the intended sketch is given by:
\small
$$P(d|\mathcal{L}) = \frac{\exp (\theta^\top \phi(\mathcal{L},d))}{\sum_{d'\in D(\mathcal{L})}\exp (\theta^\top \phi(\mathcal{L},d'))}$$
\normalsize
where $\theta \in \mathbb{R}^b$ is the vector of parameters to be learned.
We learn these parameters with  supervision from labeled training data, which consists of pairs  $(\mathcal{L}_i, h_i^*)$ where $\mathcal{L}_i $ is the English description and $h^*$ is a corresponding sketch label.
During learning, we maximize the log probability of the system generating $h^*$ regardless of derivation. In particular, given $N$ training samples, our objective function is:
\small
$$\max_\theta \log \sum_{i}^N \sum_{d: \textrm{sketch}(d)=h_i^*} P(d|\mathcal{L}_i)$$
\normalsize
Intuitively, 
the model increases the weight assigned to features for derivations that exactly match the annotated sketch.

In practice, $D(\mathcal{L})$ is a very large set of derivations, exponential with respect to the number of active lexical rules in the span.
Therefore, we use beam search to find the approximate highest-scoring derivation. That is, instead of keeping all  possible derivations for a span, we only keep a set of top-$m$ derivations $D_m(\mathcal{L})$ according to their probabilities and discard the rest. During training, we maximize the likelihood of the correct derivation with respect to this set; that is, normalizing over $D_m(\mathcal{L})$ rather than $D(\mathcal{L})$.

%% file: impl.tex
\section{Implementation}\label{sec:impl}


We have implemented our synthesis algorithm in a new tool called \toolname. In addition to the natural language description and positive/negative examples, \toolname\ takes two additional inputs, namely a time budget $t$ and a parameter $k$ that controls how many results to show to the user.  The output of \toolname\ consists of \emph{up to} $k$ regexes that satisfy the examples. Note that the actual number of regexes returned by \toolname\ may be less than $k$ due to the time budget.

\toolname\ is written in Java and  leverages a number of other existing  tools. 
First, our semantic parser is built on top of the SEMPRE  framework~\cite{berant} and leverages its existing functionalities, such as the  linguistic pre-processor. 
Second, \toolname\ makes use of the Z3 SMT solver~\cite{z3} for inferring possible values of the symbolic integers  (recall Section~\ref{sec:inferConstants}). 
Finally, \toolname\ uses the Brics automaton library~\cite{automaton} for checking whether a string is matched by a regex. 

The internal workflow of \toolname\ is as follows: First, the semantic parser generates up to $500$ derivations for the given utterance and ranks them using the machine learning model. Then, we take the top $25$ sketches produced by the parser and run $25$ instances of the PBE engine \emph{in parallel} to find a completion of each sketch that is consistent with the given examples. Then, given a value of $k$ that can be specified by users, we wait for up to $k$ PBE engine instances to complete their task  and return the synthesized regexes for those tasks that terminate within the given time budget $t$.



\vspace{-5pt}
\paragraph{Eliminating membership queries.} For every concrete regex $r$ explored by our synthesis algorithm, we need to check whether $r$ matches all positive examples and rejects all negative ones. Thus, \toolname\ ends up issuing many regular language membership queries, some of which are quite expensive in practice.
To reduce this overhead, our implementation uses various heuristics to eliminate unnecessary membership queries. For example, if we have determined that the regex {\tt Contains}($r$) does not match one of the positive examples, then we know that {\tt StartsWith}($r$) will also not match at least one of the examples. Similarly, if we have determined that the regex {\tt RepeatAtLeast}($r, 2$) does not match a positive example, we can conclude {\tt RepeatAtLeast}($r, k$) will not match the examples for any value of $k \geq 2$. Our implementation uses such ``subsumption'' heuristics to eliminate some of the redundant membership queries. 

\vspace{-5pt}
\paragraph{Eliminating redundant sketches.}
During semantic parsing, duplicate tokens in a span  lead to many redundant derivations. 
We eliminate these duplicate sketches during beam search and keep the generated derivations non-identical.

%% file: setup.tex
\section{Data Sets for Evaluation}\label{sec:setup}

To conduct our experiments, we collected two data sets, one of which is an adapted version of a data set used in \deepregex~\cite{deepregex} and another  much more challenging data set curated from StackOverflow. 

\vspace{-2pt}
\paragraph{DeepRegex data set} 

As mentioned earlier, \deepregex\ is a tool for generating regexes directly from natural language~\cite{deepregex}.  
However, to evaluate our technique on the \deepregex\ data set, we need positive and negative examples in addition to the English description. Thus, to adapt this data set to our setting, we took 200 benchmarks from this data set and asked users to provide positive and negative examples\footnote{{The details of this data set and the procedure for adapting it to our setting are described in the Appendix.}}. On average, each benchmark in this adapted \deepregex\ data set contains 4 positive and 5 negative examples.

\vspace{-2pt}
\paragraph{StackOverflow data set}

To evaluate \toolname\ on more realistic string matching tasks encountered by real-world users, we also collected a set of \emph{much more challenging} benchmarks from StackOverflow.  Specifically, we  searched StackOverflow using relevant keywords, such as \emph{``regex''}, \emph{``regular expression''}, \emph{``text validation''} etc. and retained \emph{all} benchmarks that contain  \emph{both} an English description as well as  positive \emph{and} negative examples.  Using this methodology, we obtained a total of $122$ regex-related tasks \revise{and generated the ground-truth regex by directly converting the answer on StackOverflow to our DSL.}

\vspace{-2pt}
\revise{\paragraph{Training for each data set.}  As described in Section 6.3, our semantic parser is parametrized by a vector $\theta$ that is used for assigning scores to each possible derivation. Because these parameters are learned using supervision from labeled training data, we need training data for each data set in the form of pairs of English sentences and their corresponding h-sketches. However, since the original data sets are not annotated with hierarchical sketches, we had to construct the h-sketches used for training ourselves.}

\revise{In general, the optimal h-sketch to use for training is hard to determine. On the one extreme, we can write an h-sketch that is exactly the target regex, but that would lead to poor performance of the semantic parser on the test set. On the other extreme, we can use a sketch that is completely unconstrained but that would be completely unhelpful for the PBE engine. To achieve a reasonable trade-off between these two extremes, we used the following strategy. For the \deepregex\ dataset where the target regexes are relatively small and simple, we automatically generated the h-sketch by replacing the top-level (root) operator with a hole. For example, if the target regex is ${\tt Concat}(\texttt{<num>}, \texttt{<let>})$, our h-sketch used for training would be $\hole\{\texttt{<num>}, \texttt{<let>}\}$. While this strategy worked well for the \deepregex\ dataset, it was not sufficiently fine-grained for the much more difficult StackOverflow benchmarks. Therefore, we manually constructed the h-sketches for the StackOverflow benchmarks by reading the English description and expressing its high-level structure as an h-sketch.  In many cases, our manually-written h-sketch 
faithfully captures the unambiguous parts of the English description (e.g., letter) but replaces ambiguous (or difficult to parse) fragments with holes. }

\vspace{-2pt}
\paragraph{Settings for each data set.} Recall from Section~\ref{sec:impl} that \toolname\ is parametrized by two additional inputs $t,k$ that control the time budget and number of results to display. For the easier \deepregex\ data set, we set a time-out limit of $10$ seconds and display only a single result. For the much harder StackOverflow benchmarks, we set the time budget to be $60$ seconds and display the top $5$ results. For performing comparisons, we use the same values of $t$ and $k$ across all tools and consider the benchmarks to be successfully solved if the intended regex is within the top $k$ results.

%% file: eval.tex
\section{Experimental Results}\label{sec:eval}

In this section, we describe a series of three experiments that are designed to answer the following research questions:

\begin{itemize}[leftmargin=*]
    \item \textbf{Q1:} What is the benefit of multi-modal synthesis? Does our approach work better compared to alternative approaches that use \emph{only} examples or \emph{only} natural language?
    \item \textbf{Q2:} How effective is our proposed PBE technique? In particular, how useful is sketch-guided deduction (Sec.~\ref{sec:over-under}) and  SMT-based solving of symbolic regexes (Sec.~\ref{sec:inferConstants})?
    \item \textbf{Q3:} Is \toolname\ helpful to users in constructing regular expressions for a given task? 
\end{itemize}

All  experiments   are conducted on an Intel Xeon(R) E5-1620 v3 CPU with 32GB physical memory.

\subsection{Benefits of multi-modal synthesis}

To evaluate the benefits of leveraging two different specification modalities, we compare \toolname\ against two baselines. Our first baseline is \deepregex\ which directly translates the natural language description into a regex using a sequence-to-sequence model~\cite{deepregex}.  Our second baseline is a variant of \toolname, henceforth referred to as \toolvar, that only uses positive and negative examples. In particular, \toolvar\ starts with a completely unconstrained sketch (i.e.,  single hole) and searches for a regex that satisfies the examples using the same algorithm described in \secref{synthesis}. {\footnote{\revise{As we show in the next subsection, \toolvar\ outperforms prior state-of-the-art regex PBE techniques; thus, we take  \toolvar\ as the representative state-of-the-art approach for synthesizing regular expressions purely from examples.}}


Since PBE tools are meant to be used interactively, we use the following methodology. First, we  run both \toolname\ and \toolvar\ on the initial examples in the original data set and consider synthesis to be successful if the intended regex is among those  returned by the tool. If it is unsuccessful, in the next iteration, we provide two additional examples \revise{that are guaranteed to rule out the returned incorrect regex}. We continue this process up to a maximum of four iterations.  


\begin{figure}[!t]
\begin{center}
\hspace{-20pt}
\begin{minipage}[t]{0.46\linewidth}
            \begin{center}
                \input{eval-dr-interactive-solved}        
            \end{center}
         \end{minipage}
         \hspace{25pt}
         \begin{minipage}[t]{0.46\linewidth}
            \begin{center}
                \input{eval-so-interactive-solved}  
            \end{center}
         \end{minipage}
         \caption{Number of solved benchmarks over iterations.}\figlabel{solved}
%
\hspace{-20pt}
            \begin{minipage}[t]{0.46\linewidth}
                \begin{center}
                    \input{eval-dr-interactive-time}        
                \end{center}
             \end{minipage}
             \hspace{25pt}
             \begin{minipage}[t]{0.46\linewidth}
                \begin{center}
                    \input{eval-so-interactive-time}  
                \end{center}
             \end{minipage}

             \caption{Average running time per solved benchmark over iterations.  Time for \deepregex's \emph{seq2seq} model is negligible.}\figlabel{time}
             
        \end{center}
\end{figure}

Our results are summarized in Figures 14 and 15. For each figure, the $x$-axis shows the number of iterations  and the $y$-axis shows either the number of benchmarks that can be successfully solved (\figref{solved}) or the average running time per benchmark (\figref{time}). For each figure, (A) shows results for the \deepregex\ data set and (B) is for  StackOverflow. 
The green line (with squares) corresponds to \toolname, the blue line (with circles) is  \toolvar, and the violet line (with triangles) is  \deepregex. 
Because \deepregex\ only takes natural language as input, the \deepregex\ line in \figref{solved} is flat. Furthermore, since \deepregex\ does not involve any search, its running time is negligible and not shown  in \figref{time}. 

\paragraph{DeepRegex data set.} Let us first focus on the results for the \deepregex\ data set, shown in \figref{solved} (A) and \figref{time} (A). 
Given the original examples in this data set, \toolname\ can produce the \emph{intended} regexes for 151 out of 200 benchmarks (75.5\% accuracy). 
Furthermore, \toolname\ solves up to 185 benchmarks (92.5\%) when more examples are available. 
In comparison,  \deepregex\ solves 134 benchmarks (67\%), whereas \toolvar\ solves at most 66 benchmarks (33\%). 
Furthermore, as illustrated in \figref{time} (A), using  the natural language specification also substantially speeds up the PBE engine. 

\paragraph{StackOverflow data set.} Next, we consider the StackOverflow results  shown in \figref{solved} (B).
As expected, the accuracy is much lower compared to the \deepregex\ data set, as the StackOverflow benchmarks are much more challenging.\footnote{In particular, the average number of words in a StackOverflow benchmark is 26 whereas \deepregex\ benchmarks have 12 words on average. Furthermore, the average \revise{AST node} size of the target regex is 13 for the StackOverflow data set and 5 for the \deepregex\ data set.} Thus, the  two baselines (namely, \deepregex\ and \toolvar) can only solve  3 (2.4\%) and 18 benchmarks (14.7\%) respectively, out of 122 benchmarks in total. 
In contrast, \toolname\ is able to solve up to 74 benchmarks out of 122 (60.7\%). 

\paragraph{Failure analysis for StackOverflow.} \revise{ To gain insight about cases where \toolname\ does not work well, we investigate several StackOverflow benchmarks where \toolname\ fails to synthesize the intended regex.  Among the benchmarks we inspected, we notice that the English description in many of the failure cases rely on high-level concepts such as \emph{date}, \emph{range}, etc. that our semantic parser has no knowledge of; therefore, the generated sketch does not precisely capture the  English description in most failure cases.}

\vspace{3pt}
\noindent
\fbox{\parbox{.95\linewidth}{
	{\bf Result 1}: Among $322$ regex tasks, \toolname\ solves  80\% of the benchmarks but \deepregex\ solves only 43\% and the PBE-baseline solves only 26\%.
  }
}

\subsection{Evaluation of PBE engine}

\input{eval-technique}


In this section, we describe an ablation study that allows us to quantify the impact of the pruning techniques described in Sections~\ref{sec:over-under} and~\ref{sec:inferConstants}. Specifically, in Figure~\ref{fig:eval-technique}, we plot the number of solved sketches against cumulative running time {for \toolname\ and two other baselines.} \revise{In this context, a sketch is considered as solved if the PBE engine can  find an instantiation of the sketch that is consistent with the examples. In this experiment, we evaluate the following PBE engines:} 

\begin{itemize}[leftmargin=*]
\item {\emph{\mina:}} The plot labeled \mina\ is a baseline that implements the pruning techniques described in \mina~\cite{lee}. Specifically, we adapt \mina\ to perform sketch-guided enumerative search (instead of breadth-first search) but use their pruning technique instead of the ideas proposed in Sections~\ref{sec:over-under} and~\ref{sec:inferConstants}.
\item \emph{\secondBL:} This variant uses the pruning techniques described in Section~\ref{sec:over-under} but does not leverage the symbolic regex idea introduced in Section~\ref{sec:inferConstants}. 
\item \emph{\toolname:} This corresponds to the full \toolname\ system incorporating both ideas from Sections~~\ref{sec:over-under} and~\ref{sec:inferConstants}.
\end{itemize}

As we can see from Figure~\ref{fig:eval-technique}, both  pruning techniques 
discussed in Sections~\ref{sec:over-under} and~\ref{sec:inferConstants} have a significant positive impact on the running time of the synthesizer. 


\noindent
\fbox{\parbox{.95\linewidth}{
	{\bf Result 2}: For the first 1000 sketches that can be solved by all variants, \toolname\ is around $10\times$ faster than \mina\ and $2.5\times$ faster than \secondBL.
  }
}

\subsection{User study}

 To further evaluate whether \toolname\ helps users complete regex-related tasks, we conducted a user study involving 20 participants, 5 of whom are professional software engineers and 15 of whom are computer science students. 
 Each participant was provided with $6$ regex tasks randomly sampled from the StackOverflow benchmarks, \revise{regardless of whether \toolname\ is able to solve that benchmark or not}. \revise{Then, we provided each participant with the original task description  in the StackOverflow post (including both the English decription and the examples) and asked them to solve exactly a (randomly selected) half of the examples using \toolname \ and the remaining half without \toolname. } For both set-ups, the users had a total of $15$ minutes to work on each setting (with \toolname\ or without \toolname).  More details about our user study set-up can be found in the appendix}.

 For the set-up involving our tool, participants  were just provided with the tool and educated about how to use it, but they were not required to use \toolname\ in any specific way. \revise{Furthermore, while the participants were  provided with the original StackOverflow post describing the task, they were free to modify both the English description and the examples as they saw fit.}
 

\paragraph{Results.}  In the set up where participants did not have access to \toolname, they correctly solved  $28.3\%$ of the benchmarks (i.e., produced the intended regex) in the given time limit. In contrast, when they had access to \toolname, success rate went up to $73.3\%$.  As standard when doing user studies, we ran a $1$-tailed $t$-test to evaluate whether our results are statistically significant. The $p$-value for this test is less than $0.0000001$. Thus, our user study provides firm evidence that the proposed technique makes it easier for users to write regexes.

\vspace{-2pt}
\revise{\paragraph{Failure case analysis.} To gain some insight about failure cases in the user study, we manually inspected those scenarios in which users were not able to successfully use \toolname\ to derive the correct regex.  Overall, we found two main root causes for failure. First, because our tasks are randomly selected from the StackOverflow benchmarks, \toolname\ 
times out on some tasks and is unable to produce any regex. In such cases, solving the benchmark with \toolname\ is no different from solving the benchmark without \toolname. Another main reason for failure is the inherent ambiguity in the StackOverflow post. That is, even with the provided examples, there may be multiple ways to interpret the question, so the users sometimes take one interpretation over the intended one and therefore select the wrong regex. (Note that users in our study were not provided with follow-up questions and discussions in the original StackOverflow post.)}

\vspace{-2pt}
\revise{\paragraph{Disclaimers.} While we believe that our user study results provide some preliminary evidence of the potential usefulness of a \toolname-like approach, our results are not intended to be a scientific study of the use of \toolname ``in the wild'' for the following reasons. First, the majority of the participants in our user study are computer science students from the same university. Second, in order to allow a fair comparison between the two approaches across all participants, our tasks are taken from StackOverflow posts as opposed to real-world tasks that the participants \emph{themselves} want to complete. }



\vspace{3pt}
\noindent
\fbox{\parbox{.95\linewidth}{
	{\bf Result 3}: \revise{For the particular setup evaluated in our small user study, } \toolname\ users are $2\times$ more likely to construct the correct regex using \toolname\ within a given time budget. 
  }
}

%% file: eval-dr-interactive-solved.tex
    \begin{tikzpicture}[scale=0.55]
        \begin{axis}[
            ymax=240,
            y=0.03cm,
            x=40,
            legend cell align = left,
            legend pos = south west,
            legend style = {
                at={(0.25,0.98)},
                legend columns=1,
                anchor=north
            },
            xlabel style={yshift=1mm},
            ylabel = \# Solved Benchmarks,
            xlabel = \# of Iterations,
            xmax = 4.5
        ]
        \legend{\toolname, \toolvar, \deepregex}
        \addplot[
            smooth,
            line width=0.4mm,
            mark=square*,
            color=teal
        ] coordinates {
            (0, 151)
            (1, 179)
            (2, 184)
            (3, 184)
            (4, 185)

        };
        \addplot[
            smooth,
            line width=0.4mm,
            mark=*,
            color=blue
        ] plot coordinates {
            (0, 45)
            (1, 61)
            (2, 65)
            (3, 66)
            (4, 66)

        };
        \addplot[
            smooth,
            line width=0.5mm,
            mark=triangle*,
            color=violet
        ] plot coordinates {
            (0, 134)
            (1, 134)
            (2, 134)
            (3, 134)
            (4, 134)

        };
        \end{axis}
    \end{tikzpicture}
    
    \captionsetup{font={small}}
    	\caption*{ \hspace{25pt} (A) \deepregex\ data set} 

%% file: eval-so-interactive-solved.tex
    \begin{tikzpicture}[scale=0.55]
        \begin{axis}[
            ymax=122,
            y=0.048cm,
            x=40,
            legend cell align = left,
            legend pos =  north west,
            xlabel style={yshift=1mm},
            ylabel = \# Solved Benchmarks,
            xlabel = \# of Iterations,
            xmax = 4.5
        ]
        \legend{\toolname, \toolvar, \deepregex}
        \addplot[
            smooth,
            line width=0.4mm,
            mark=square*,
            color=teal
        ] coordinates {
            (0, 53) 
            (1, 66)
            (2, 68)
            (3, 71)
            (4, 74)

        };
        \addplot[
            smooth,
            line width=0.4mm,
            mark=*,
            color=blue
        ] plot coordinates {
            (0, 8)
            (1, 14)
            (2, 17)
            (3, 17)
            (4, 18)
        };

        \addplot[
            smooth,
            line width=0.5mm,
            mark=triangle*,
            color=violet
        ] plot coordinates {
            (0, 3)
            (1, 3)
            (2, 3)
            (3, 3)
            (4, 3)
        };
        \end{axis}
    \end{tikzpicture}
    
    \captionsetup{font={small}}
    \caption*{\hspace{15pt} (B) StackOverflow data set}

%% file: eval-dr-interactive-time.tex
    \begin{tikzpicture}[scale=0.55]
        \begin{axis}[
            ymax=4,
            y=1.5cm,
            x=40,
            legend cell align = left,
            legend pos = outer north east,
            legend style = {
                at={(0.25,0.98)},
                legend columns=1,
                anchor=north
            },
            xlabel style={yshift=1mm},
            ylabel = Avg Time (s),
            xlabel = \# of Iterations,
            xmax = 4.5
        ]
        \legend{\toolname, \toolvar}
        \addplot[
            smooth,
            line width=0.4mm,
            mark=square*,
            color=teal
        ] coordinates {
            (0, 0.09)
            (1, 0.14)
            (2, 0.17)
            (3, 0.17)
            (4, 0.17)

        };
        \addplot[
            smooth,
            line width=0.4mm,
            mark=*,
            color=blue
        ] plot coordinates {
            (0, 0.61)
            (1, 1.57)
            (2, 1.86)
            (3, 1.91)
            (4, 1.91)

        };
        \end{axis}
    \end{tikzpicture}
    
    \captionsetup{font={small}}
    \caption*{\hspace{25pt} (A) \deepregex\ data set}

%% file: eval-so-interactive-time.tex
    \begin{tikzpicture}[scale=0.53]
        \begin{axis}[
            ymax=20,
            ymin=3,
            y=0.40cm,
            x=40,
            legend cell align = left,
            legend pos = outer north east,
            legend style = {
                at={(0.25,0.98)},
                legend columns=1,
                anchor=north
            },
            xlabel style={yshift=1mm},
            ylabel =  Avg Time (s),
            xlabel = \# of Iterations,
            xmax = 4.5
        ]
        \legend{\toolname, \toolvar}
        \addplot[
            smooth,
            line width=0.4mm,
            mark=square*,
            color=teal
        ] coordinates {
            (0, 3.7)
            (1, 6.1)
            (2, 5.9)
            (3, 6.7)
            (4, 6.5)

        };
        \addplot[
            smooth,
            line width=0.4mm,
            mark=*,
            color=blue
        ] plot coordinates {
            (0, 7.38)
            (1, 12.52)
            (2, 14.59)
            (3, 14.59)
            (4, 14.59)

        };
        \end{axis}
    \end{tikzpicture}
    
    \captionsetup{font={small}}
    \caption*{\hspace{15pt} (B) StackOverflow data set}

%% file: related.tex
\section{Related Work}

In this section, we review prior work on program synthesis from examples and natural language.

\vspace{-6pt}

\paragraph{Learning regexes from examples} 

There is a large body of prior research  on learning regular expressions from positive and negative examples~\cite{alquezar94, firoiu98, ANGLUIN1978337, gold1978, rivest1989, DFA1, DFA2}, including Angluin's  well-known $L^*$ algorithm for active learning of regular expressions~\cite{angluin1987}. In this setting, a regular language is represented by an oracle that can answer membership queries, check for equivalence, and provide counterexamples. \revise{While  these algorithms can learn the target  language in polynomial time (with respect to the minimal DFA), they tend to require orders of magnitude more examples compared to our approach. For instance, for the simple regex {\tt [A-Za-z]+}, an implementation of the $L^*$ algorithm asked 679 queries before it synthesized the correct regex whereas \toolvar\ was able to synthesize the desired regex using 8 examples.}



More recent work that is  closely related to our approach is \textsc{AlphaRegex} ~\cite{lee} which also performs top-down enumerative search and uses over- and under-approximations to prune the search space.  However, \textsc{AlphaRegex} can only synthesize regexes over a binary alphabet and does not utilize natural language. In contrast to  \textsc{AlphaRegex}, we use hierarchical sketches for both guiding the search and pruning infeasible regexes. Additionally, our method uses symbolic regexes and SMT-based reasoning to further prune the search space.  Another related tool is \textsc{RFixer}, which performs repair on regular expressions~\cite{rfixer}. Rather than performing synthesis from scratch, \textsc{RFixer} modifies a given regex to be consistent with the provided examples and also uses techniques similar to \textsc{AlphaRegex} to prune the search space.

\vspace{-6pt}
\paragraph{Learning regexes from language} 

There has been recent interest in automatically generating regexes from natural language.  
For example, \citet{KB13} build a dependency parser for translating natural language text queries into regular expressions. Their technique is built on top of a combinatory categorical grammar and  utilizes semantic unification to improve training. 
Other work in this space uses seq2seq models to predict regular expressions from English descriptions~\cite{deepregex,semregex}.  However, these techniques do not utilize examples and attempt to directly translate natural language into a regex rather than a sketch.

\vspace{-6pt}
\paragraph{Multi-modal synthesis} 
There has been recent interest in synthesizing string manipulation programs from both natural language and examples. 
For instance, \citet{manshadi13} propose a PBE system that leverages natural language in order to deduce the correct program more often and faster. Specifically, they use natural language to construct a so-called \emph{probabilistic version space} and apply this idea to string transformations expressible in a subset of the FlashFill DSL~\cite{flashfill}.
\citet{Raza} also use propose combining natural language and examples but do so in a very different way. Specifically,  they try to decompose the English description into constituent concepts and then ask the user to provide examples for each concept in the decomposition.


Similar to our approach, there have also been recent proposals to combine natural language and examples using a sketching-based approach. For instance, \cite{sketchadapt} provides a framework for generating program sketches from any type of specification, which can also involve natural language. Specifically, they first use an LSTM to generate a distribution over program sketches and then try to complete the sketch  using a generic sketch completion technique based on breadth-first enumeration. Another related effort in this space is the \textsc{Mars} tool which also utilizes natural language and examples \cite{mars}.  In contrast to our technique, they derive soft constraints from natural language and utilize a MaxSMT solver to perform synthesis. In addition, \textsc{Mars} targets data wrangling applications rather than regexes.

\vspace{-6pt}
\paragraph{PBE and sketching} 

Similar to this work, several recent PBE techniques combine top-down enumerative search with lightweight deductive reasoning to significantly prune the search space~\cite{lambda2, neo, morpheus, aws13, le14, yaghmazadeh16, osera15}.  Our method also bears similarities to sketching-based approaches~\cite{solarthesis} in two ways: First, we generate some sort of program sketch from the natural language description. However, in contrast to prior work, our sketches are hierarchical in nature, and the holes in the sketch represent arbitrary regexes rather than constants. Second,  we use a constraint-solving approach to infer constants in a symbolic regex. However, compared to most existing techniques~\cite{jha2010, gulwani2011, ahish2015, emina}, we use constraint solving as a way to rule out infeasible integer constants rather than directly solving for them.



\vspace{-6pt}
\paragraph{Program synthesis from NL} 

Beyond regexes, there have also been proposals for performing program synthesis directly from natural language~\cite{sqlDL, neuralprogram, lin2018}. Such techniques have been used to generate SQL queries~\cite{sqlizer,sqlDL}, ``if-this-then-that recipes'' \cite{ifttt},  spreadsheet formulas~\cite{nlyze}, bash commands~\cite{lin2018}, and Java expressions~\cite{viktor}.  Our technique is particularly similar to {\sc SQLizer}~\cite{sqlizer} in that we also infer a sketch from the natural language description. However, unlike our approach, {\sc SQLizer} does not utilize examples and populates the sketch using a different technique called \emph{quantitative type inhabitation}~\cite{type}.

%% file: conc.tex
\section{Conclusions and Future Work}

In this paper, we presented a new method, and its implementation in a tool called \toolname, to synthesize regular expressions from a combination of examples and natural language.  The key idea underlying our approach is to generate a hierarchical sketch from the English description and use the hints embedded in this sketch to guide both search and deduction.  We evaluated our approach on 322 regexes obtained from two different sources and showed that our approach can successfully synthesize the intended regex in 80\% of the cases within four user interaction steps. In comparison,  a state-of-the-art tool  that uses  only natural language can  solve 43\% of these benchmarks and an example-only baseline can solve only 26\%. We also performed an
evaluation of our PBE engine and showed that \toolname\ is an order of magnitude faster compared to \mina, a state-of-the-art PBE tool for regex synthesis. 




\revise{In future work, we are interested in exploring a multi-modal \emph{active learning} approach to synthesizing regular expressions. In our current work, \toolname\ produces top-$k$ results that satisfy the examples, but it is up to the user to inspect these results and provide more examples as needed. However, we believe  it would be  beneficial to develop a regex synthesis tool that would ask  the user membership queries to disambiguate between multiple different solutions that are consistent with the examples.  We are also interested in semantic parsing or other NLP techniques that might generate helpful feedback to users in cases where the generated sketch is too coarse. Finally, we plan to explore the use of the proposed synthesis methodology in application domains beyond regular expressions.}



%% file: appendix.tex
\begin{appendices}
    \section{Proofs}
    \begin{lemma}{\bf (Correctness of \textnormal{\textsc{Approximate}} in \figref{approx-sketch})}
        Given an h-sketch $\sketch$, \textnormal{\textsc{Approximate}} constructs $\langle  \overapprox, \underapprox  \rangle$ where $\overapprox$ over-approximates $\partialprog$ and $\underapprox$ under-approximates $\sketch$. That is, we have 
        \begin{center}
        $\textnormal{(i)} \  \forall s. \   (\exists r \in \semantics{\sketch}.\ \textsf{Match}(r, s)) \Rightarrow \textsf{Match}(\overapprox, s)$ \\
        $\textnormal{(ii)} \  \forall s. \   \textsf{Match}(\underapprox, s) \Rightarrow (\forall r \in \semantics{\sketch}.\ \textsf{Match}(r, s))$ 
        \end{center}
        \lemmalabel{approx2}
    \end{lemma}

    \begin{proof}
    
        We prove this by structural induction on $\sketch$, as follows:
        \begin{itemize}
            \item Base case: $\sketch$ is of the form $\hole_1\{\sketch_1\}$. By the rule (1) in \figref{expandalg}, we know that this hole must be instantiated by $\sketch_1$. Therefore the over and under approximation for this $\sketch$ is the over and under approximation of $\sketch_1$
            \item Base case: $\sketch$ is of the form $\hole_d\{\sketch_1, \ldots, \sketch_m\}$. This case is trivial from the definition of over and under approximation.
            \item Base case: $\sketch$ is of the form of a concrete regex $r$. This case is trivial because we don't have to do any over and under approximation. 
            \item Inductive case: $\sketch = \hole_1\{\sketch_1, \ldots,\sketch_m\}$. By induction hypothesis, the approximation for \\ $\hole_1\{\sketch_2, \ldots, \sketch_m\}$ is $\langle o', u' \rangle$, we now prove (i) and (ii) holds for $\langle \texttt{Or} (o, o'), \texttt{And}(u, u') \rangle$, where $\langle o, u \rangle$ is the approximation for $\sketch_1$.
            \begin{itemize}
                \item (i) holds for the over-approximation $\texttt{Or} (o, o')$. Given a string $s$ matched by a regex $r$ instantiated from $\sketch$, from the semantic of $\sketch$, we know that such $r$ is either instantiated from $\sketch_1$ or $\hole_1\{\sketch_2, \ldots, \sketch_m\}$. From the inductive hypothesis, if $r$ is instantiated from $\sketch_1$, then $\textsf{Match}(o,s)$ is true, and if $r$ is instantiated from $\hole_1\{\sketch_2, \ldots, \sketch_m\}$, $\textsf{Match}(o',s)$ is true. From the semantic of $\texttt{Or}$ operator, if $r \in \semantics{\sketch}$ and $r$ matches $s$, then $\texttt{Or} (o, o')$ is true. 
                \item (ii) holds for the under-approximation $\texttt{And}(u,u')$. Suppose the under-approximation matches a string $s$. From the semantics of the $\texttt{And}$ operator, we know that both $u$ and $u'$ match $s$. From the inductive hypothesis, we know that for all $r'$ instantiated by $\hole_1\{\sketch_2, \ldots, \sketch_m\}$, $\textsf{Match}(r',s)$ is true; also from the base case we know that all $r''$ instantiated by $\sketch_1$, $\textsf{Match}(r'',s)$ is true. From the semantics of $\sketch$, we know that all $r \in \semantics{\sketch}$, $r \in {r'}$ or $r \in {r''}$. Therefore $\forall r \in \semantics{\sketch}.\ \textsf{Match}(r,s)$ is true. 
    
            \end{itemize}     
            \item Inductive case: $\sketch = \texttt{f}(\sketch_1, \ldots, \sketch_n) \text{ where } \texttt{f} \in \mathcal{F}_n$. By induction, for each $\sketch_i (i = 1, \ldots, n)$, $\langle o_i, u_i \rangle$ satisfies (i) and (ii). Now we show that $\langle o, u \rangle$ satisfies (i), (ii) as well, by considering all possibilities of operator $\texttt{f}$
            \begin{itemize}
                \item $\sketch = \texttt{StartsWith}(\sketch_1)$.
                \begin{enumerate}
                    \item We first prove that $o$ satisfies (i). For any string $s$, suppose there exists a regex $\prog \in \semantics{\sketch}$ such that $\prog = \texttt{StartsWith}(\prog_1)$ such that we have $\textsf{Match}(\prog, s)$. From the semantic of h-sketch from \figref{sketchsemantics}, we know that $\prog_1 \in \semantics{\sketch_1}$. By induction, we know that (i) holds for $\sketch_1$. Thus, we have $\textsf{Match}(o_1, s_1)$ implies $\textsf{Match}(o,s) $ since $o$ is $\texttt{StartsWith}(o_1)$ according to rule (4) and $s_1$ is a prefix of $s$. Therefore $o$ satisfies (i).
                    \item We now prove that $u$ satisfies (ii). For any string $s$, suppose $\textsf{Match}(u,s)$ is true, then there exist a string $s_1$ such that $\textsf{Match}(u_1, s_1)$ and $u = \\\texttt{StartsWith}(u_1)$. From the inductive hypothesis, we know that $\textsf{Match}(u_1, s_1)$ holds for any $\prog_1 \in \semantics{\sketch_1}$. Now consider any regex $\prog \in \semantics{\partialprog}$, because we have $\prog = \texttt{StartsWith}(\prog_1)$, $s_1$ is a prefix of $s$ and $\textsf{Match}(\prog_1, s_1)$, we have $\textsf{Match}(\prog, s)$.
                \end{enumerate}
                \item $\sketch = \texttt{f}(\sketch_1, \ldots, \sketch_n) \text{ where } \texttt{f} \in \mathcal{F}_n $ is \texttt{Contains}, \texttt{EndsWith}, \texttt{Concat}, \texttt{And}, \texttt{Or}, \texttt{Optional}, \\\texttt{KleeneStar}. The proof is similar to that for \texttt{StartsWith}.
                \item $\sketch = \texttt{Not}(\sketch_1)$.
                \begin{enumerate}
                    \item We first prove $o$ satisfies (i). Given a string $s$, suppose there exists a concrete regex $\prog \in \semantics{\sketch}$ such that $\textsf{Match}(\prog, s)$ is true. From the semantic of $\texttt{Not}$, we know that $\neg \textsf{Match}(\prog_1, s)$, where $\prog_1 \in \semantics{\sketch_1}$. From the induction hypothesis, we know that $\\\neg \textsf{Match}(u_1, s)$, where $u_1$ is the under-approximation for $\sketch_1$, therefore, $o = \texttt{Not}(u_1)$ is true. Hence $\textsf{Match}(o, s)$ is true. 
                    \item We then prove $u$ satisfies (ii). For any string $s$, suppose $\textsf{Match}(u,s)$ is true. Since we have $u = \texttt{Not}(o_1)$, where $o_1$ is the over-approximation for $\sketch_1$, $\textsf{Match}(o_1, s)$ is false (from the semantics of $\texttt{Not}$). From the inductive hypothesis, we know that for any $\prog_1 \in \semantics{S_1}$, $\neg \textsf{Match}(r_1,s)$. Therefore, for any $\prog \in \semantics{S}$ where $\sketch = \texttt{Not}(\sketch_1)$, we have $\textsf{Match}(\prog, s)$. Therefore, $u$ satisfies (ii).
                \end{enumerate}
            \end{itemize}
            \item Inductive case: $\sketch = \texttt{g}(\sketch_1, \kappa_1, \ldots, \kappa_n) \text{ where } \texttt{g} \in \mathcal{G}_n$. Since $u = \bot$ clearly satisfies (ii), here we only prove that $o$ satisfies (i). Also since the cases for \texttt{RepeatAtLeast} and \texttt{RepeatRange} are similar, here we only prove for \texttt{Repeat}. Given any string $s$, suppose there exist $\prog \in \semantics{\sketch}$ such that $\prog = \texttt{Repeat}(\prog_1, \kappa)$ and $\textsf{Match}(\prog, s)$ is true, where $\prog_1 \in \semantics{S_1}$. From the semantic of the \texttt{Repeat} operator, we then know $\textsf{Match}(\prog_1, s_1)$, $\ldots$, $\textsf{Match}(\prog_1, s_\kappa)$ are true, where $s$ is the concatenation of $s_1, \ldots, s_\kappa$. From inductive hypothesis, we know $o_1$ matches $s_1, \ldots, s_\kappa$. From the semantic of \texttt{RepeatAtLeast}, we know that $o = \texttt{RepeatAtLeast}(o_1, 1)$ matches $s$, i.e. $\textsf{Match}(o, s)$ is true. Therefore $o$ satisfies (i).
        
        \end{itemize}

    \end{proof}

    \begin{theorem}{\bf (Correctness of \textnormal{\textsc{Approximate}} in \figref{approxrules})}
    Given a partial regex $\partialprog$,  \textnormal{\textsc{Approximate}} constructs $\langle  \overapprox, \underapprox  \rangle$ where $\overapprox$ over-approximates $\partialprog$ and $\underapprox$ under-approximates $\partialprog$. That is, we have 
    \begin{center}
    $\textnormal{(i)} \  \forall s. \   (\exists r \in \semantics{\partialprog}.\ \textsf{Match}(r, s)) \Rightarrow \textsf{Match}(\overapprox, s)$ \\
    $\textnormal{(ii)} \  \forall s. \   \textsf{Match}(\underapprox, s) \Rightarrow (\forall r \in \semantics{\partialprog}.\ \textsf{Match}(r, s))$ 
    \end{center}
    \theoremlabel{approx1} 
    \end{theorem}

    \begin{proof}
    We prove this by structural induction on $\partialprog$, as follows. 
    \begin{itemize}[leftmargin=*]
    \item 
    Base case: $\partialprog$ is an h-sketch $\sketch$. 
    This case is trivial according to rule (1) and \lemmaref{approx2}. 
    \item 
    Inductive case: $\partialprog$ is of the form $\texttt{f}(\partialprog_1, \dots, \partialprog_n)$ where $\texttt{f} \in \mathcal{F}_n$. 
    By induction, for each $\partialprog_i \ (i = 1, \dots, n)$, $\langle \overapprox_i, \underapprox_i \rangle$ satisfies (i) and (ii). 
    Now, we show that $\langle \overapprox, \underapprox \rangle$ satisfies (i) and (ii) as well, by considering all possibilities of operator \texttt{f}. 
    \begin{itemize}
    \item 
    $\partialprog = \texttt{StartsWith}(\partialprog_1)$. 
    We first prove that $\overapprox$ satisfies (i). 
    For any string $s$, suppose there exists a regex $\prog = \texttt{StartsWith}(\prog_1)$ such that we have $\textsf{Match}(\prog, s)$.  Then, we have $\textsf{Match}(\prog_1, s_1)$ where $s_1$ is some prefix of $s$. By induction, we know that (i) holds for $\partialprog_1$. Thus, we have $\textsf{Match}(\overapprox_1, s_1)$, which implies $\textsf{Match}(\overapprox, s)$ since $\overapprox$ is $\\\texttt{StartsWith}(\overapprox_1)$ according to rule (2) and $s_1$ is a prefix of $s$. Therefore, $\overapprox$ satisfies (i). 
    Now, we prove that $\underapprox$ satisfies (ii). 
    For any string $s$, suppose we have $\textsf{Match}(\underapprox, s)$. Since we have $\underapprox = \texttt{StartsWith}(\underapprox_1)$ according to rule (2), we have $\textsf{Match}(\underapprox_1, s_1)$ for some prefix $s_1$ of $s$. By induction, we know that $\textsf{Match}(\prog_1, s_1)$ holds for any $r_1 \in \semantics{\partialprog_1}$. Now consider any regex $\prog \in \semantics{\partialprog}$. Because we have $\prog = \textsf{StartsWith}(\prog_1)$, $s_1$ is a prefix of $s$ and $\textsf{Match}(\prog_1, s_1)$, we have $\textsf{Match}(\prog, s)$. Therefore, $\underapprox$ satisfies (ii). 
    \item 
    $\partialprog = \texttt{f}(\partialprog_1, \dots, \partialprog_n)$ where $\texttt{f} \in \mathcal{F}_n$ is \texttt{Contains}, \texttt{EndsWith}, \texttt{Concat}, \texttt{And}, \texttt{Or}, \texttt{Optional}, or \texttt{KleeneStar}. 
    The proof is similar to that for \texttt{StartsWith}. 
    \item 
    $\partialprog = \texttt{Not}(\partialprog_1)$. 
    We first prove that $\overapprox$ satisfies (i). 
    For any string $s$, suppose there exists a regex $\prog = \texttt{Not}(\prog_1)$ such that we have $\textsf{Match}(\prog, s)$. Then, we have $\neg\textsf{Match}(\prog_1, s)$. By induction, we know that (ii) holds for $\partialprog_1$. Thus, we have $\neg\textsf{Match}(\underapprox_1, s)$, or in other words, $\textsf{Match}(\texttt{Not}(\underapprox_1), s)$. This implies $\textsf{Match}(\overapprox, s)$ since $\overapprox$ is $\texttt{Not}(\underapprox_1)$ according to rule (3). 
    Now we prove that $\underapprox$ satisfies (ii). 
    For any string $s$, suppose we have $\textsf{Match}(\underapprox, s)$. Since we have $\underapprox = \texttt{Not}(\overapprox_1)$ according to rule (3), we have $\textsf{Match}(\texttt{Not}(\overapprox_1, s))$, or in other words, $\neg\textsf{Match}(\overapprox_1, s)$. By induction, we know that (i) holds for $\partialprog_1$. That is, for any regex $\prog_1 \in \semantics{\partialprog_1}$ we have $\neg\textsf{Match}(\prog_1, s)$. Therefore, for any $\prog \in \semantics{\partialprog}$ where $\partialprog = \texttt{Not}(\partialprog_1)$, we have $\textsf{Match}(\prog, s)$. Therefore, $\underapprox$ satisfies (ii). 
    \end{itemize}
    \item 
    Inductive case: $\partialprog$ is of the form $\texttt{g}(\partialprog_1, k_1, \dots, k_n)$ where $\texttt{g} \in \mathcal{G}_{n+1}$ and $k_i \in \mathbb{Z}^{+}$. 
    The proof is similar to that for \texttt{StartsWith}. 
    \item 
    Inductive case: $\partialprog$ is of the form $\texttt{g}(\partialprog_1, \kappa_1, \dots, \kappa_n)$ where $\texttt{g} \in \mathcal{G}_{n+1}$ and $\kappa_i$ is a \emph{symbolic integer}. 
    Since $\underapprox = \bot$ clearly satisfies (ii), here we only prove that $\overapprox$ satisfies (i). In particular, we prove $\overapprox$ satisfies (i) for \texttt{Repeat} and the proofs for \texttt{RepeatAtLeast} and \texttt{RepeatRange} are similar. 
    For any string $s$, suppose there exists a regex $\prog = \texttt{Repeat}(\prog_1, k)$ such that we have $\texttt{Match}(\prog, s)$. Then we have $\textsf{Match}(\prog_1, s_1), \dots,\\ \textsf{Match}(\prog_1, s_k)$ where $s$ is the concatenation of $s_1, \dots, s_k$. By induction, we know that $\textsf{Match}(\overapprox_1, s_1), \dots, \textsf{Match}(\overapprox_1, s_k)$, which implies $\textsf{Match}(\overapprox, s)$ since we have $\\\overapprox = \texttt{RepeatAtLeast}(\overapprox_1, 1)$ according to rule (4). Therefore, $\overapprox$ satisfies (i). 
    \end{itemize}
    \end{proof}

    \begin{theorem}{\bf (Correctness of \textnormal{\textsc{InferConstants}} in \figref{symint})}
    Suppose given a partial regex $\partialprog$, positive examples $\posexamples$ and negative examples $\negexamples$, \textnormal{\textsc{InferConstants}} returns $\exps$. 
    Then, for any concrete regex $\prog \in \semantics{\partialprog}$ that is consistent with $\posexamples$ and $\negexamples$, we have $\prog \in \exps$. 
    \end{theorem}

    \begin{proof}

    Let the set of concrete regexes represented by the state $(\partialprog, \phi)$ be $\semantics{\partialprog}_\phi$. 

    Suppose the constraint returned by the \textsc{Encode} procedure be $(\phi,x )$.
    At line $2$, we construct a new constraint $\psi$ by conjunction all the $\phi[len(s)/x]$ where each $s \in \posexamples$. 
    Therefore, from \theoremref{encode}, we know that any concrete regex $r \in \semantics{\partialprog}$ that is consistent with $\posexamples$ and $\negexamples$, $r \in \semantics{\partialprog_0}_\psi$.

    We show that (1) at the end of each iteration, for any regex $r \in \semantics{\partialprog}$ that is consistent with positive and negative examples, $r$ is either in $\exps$ or in $\semantics{P'}, \text{for any } (P',\phi) \in worklist$.
    \begin{itemize}
        \item Base case: iteration = 1, in this case the state pulled from the worklist is $(\partialprog_0, \psi)$. If $\psi$ is UNSAT, we know that none of the program defined by $\semantics{P_
        0}_\psi$ satisfy the examples, and therefore overall $P$ does not contain any correct regex that is consistent with positive and negative examples from the definition of $\psi$.
        For  $\psi$ that is satisfiable, if $\partialprog'$ is concrete then it is trivial that $\partialprog' \cup \semantics{P_0}_{\psi \wedge \kappa \neq \sigma[\kappa]}$ still contains all the $\prog \in \semantics{P}$. If $\partialprog'$ is infeasible, from the soundness of the $\textsc{Infeasible}$ procedure we know that none of the $\prog \in \semantics{P}$ such that $ \prog \in \semantics{P'}$. Therefore, all the correct $\prog \in \semantics{P_0}_{\psi \wedge \kappa \neq \sigma[\kappa]}$. And if $\partialprog'$ is feasible, notice that $\semantics{P_0}_{\psi \wedge \kappa \neq \sigma[\kappa]} \cup \semantics{P'}_{\psi[\kappa \vartriangleleft \sigma[\kappa]]} = \semantics{P_0}_{\psi}$, therefore (1) still holds. 
        \item Inductive case: Suppose for $iteration = 2, \ldots , n$, (1) all holds, we now prove that (1) holds for the $n+1^{th}$ iteration. Let the state pulled from the worklist at this iteration be $(\partialprog_n, \phi_n)$.  If $\phi_n$ is UNSAT, then we know $\semantics{\partialprog_n}_{\phi_n}$ is a empty set. From the inductive hypothesis (1) holds for the $n^{th}$ iteration and therefore (1) still holds for $n+1$ iteration in this case. The argument for proving other cases are similar as the base case. 
    \end{itemize}

    We now show that (2) the worklist algorithm will exhaust all the possible assignments of $\psi$. Observe from line 8-12, at each iteration we replace each state $(\partialprog, \phi)$ with either a state that is more restrictive by blocking one possible assignment for one symbolic integer, or reduce the number of symbolic integer in $\partialprog$ by one while the possible assignments defined by $\phi$ is the same for rest of the symbolic variables. Also since that the total number of possible assignment for $\partialprog$ defined by the constraint $\psi$ is finite, and the number of symbolic integer allowed is finite. Eventually, this algorithm will exhaust all the possible assignments of symbolic integer of program $\partialprog$ constraint on $\psi$.

    Combining (1) and (2),  we know that  the worklist will terminates (i.e. the worklist set is empty) and any correct regex $\prog \in \semantics{P}$ is either in the $\exps$ or in $\semantics{P'}_\phi$ for any  $(\partialprog', \phi) \in \text{worklist}$, we prove that for any $\prog \in \semantics{P}$ that is consistent with 
    $\posexamples$ and $\negexamples$, $\prog \in \exps$, when \textsc{InferConstants} returns $\exps$.
    \end{proof}

    \begin{theorem}{\bf (Correctness of \textnormal{\textsc{Encode}} in \figref{constructrules})}
    Suppose \textnormal{\textsc{Encode}} returns $(\phi, x)$ for a given symbolic program $\partialprog$ with $n$ symbolic integers $\kappa_1, \dots, \kappa_n$. 
    Then given a string $s$, if regex $\partialprog[\kappa_1 \annot k_1, \dots, \kappa_n \annot k_n]$ matches $s$ (where $k_i \in [1, \textsf{MAX}], i = 1, \dots, n$), 
    we have $\kappa_1 = k_1, \dots, \kappa_n = k_n$ is a satisfying assignment of $\phi[len(s) / x]$. 
    \end{theorem}
    \theoremlabel{encode}

    \begin{proof}
    We proof this by structural induction on $\partialprog$, as follows. 
    \begin{itemize}[leftmargin=*]
    \item 
    Base case: $\partialprog$ is a character class $c$. 
    This holds obviously since the length of any string that is matched by $c$ is 1. 
    \item 
    Inductive case: $\partialprog$'s root is annotated with an operator. 
    \begin{itemize}
    \item 
    The operator is $\texttt{f} \in \mathcal{F}_n$. Here, we only show how to prove the case where $\texttt{f}$ is \texttt{Concat} (other cases are similar). 
    Given symbolic program $\partialprog = \texttt{Concat}(\partialprog_1, \partialprog_2)$ with $\kappa_1, \dots, \kappa_n$, suppose $\textsc{Encode}(\partialprog)$ returns $(\phi, x)$. Now, we show that, if regex $\partialprog[\kappa_1 \annot k_1, \dots, \kappa_n \annot k_n]$ matches string $s$, then $\kappa_1 = k_1, \dots, \kappa_n = k_n$ is a satisfying assignment of $\phi[len(s) / x]$. Without loss of generality, let us assume $\partialprog_1$ uses $\kappa_1, \dots, \kappa_m$ and $\partialprog_2$ uses $\kappa_{m+1}, \dots, \kappa_n$. We also  assume $\partialprog_1[\kappa_1 \annot k_1, \dots, \kappa_m \annot k_m]$ matches $s_1$ and $\partialprog_2[\kappa_{m+1} \annot k_{m+1}, \dots, \kappa_n \annot k_n]$ matches $s_2$ where $s$ is a concatenation of $s_1$ and $s_2$. Since $\textsc{Encode}(\partialprog_i)$ returns $(\phi_i, x_i)$, by induction we know that $\kappa_1 = k_1, \dots, \kappa_m = k_m$ is a satisfying assignment of $\phi_1[len(s_1) / x_1]$ and $\kappa_{m+1}, \dots, \kappa_n$ is a satisfying assignment of $\phi_2[len(s_2) / x_2]$. Now we show that $\kappa_1 = k_1, \dots, \kappa_n = k_n$ is a satisfying assignment of $\phi[len(s) / x]$ where $\phi$ is $\exists x_1, x_2. \ x = x_1 + x_2 \wedge \phi_1 \wedge \phi_2$. This obviously holds because we have $len(s) = len(s_1) + len(s_2)$, $\phi_1[len(s_1) / x_1] = true$ and $\phi_2[len(s_2) / x_2] = true$.  
    \item 
    The operator is $\texttt{g} \in \mathcal{G}_n$. Here, we only show how to prove the case where $\texttt{g}$ is \texttt{Repeat} (other cases are similar). 
    Given a symbolic program $\partialprog = \texttt{Repeat}(\partialprog_1, \kappa)$ where $\partialprog_1$ has symbolic integers $\kappa_1, \dots, \kappa_n$, suppose $\textsc{Encode}(\partialprog)$ returns $(\phi, x)$. Now we show that if regex $\partialprog[\kappa \annot k, \kappa_1 \annot k_1, \dots, \kappa_n \annot k_n]$ matches string $s$, then $\kappa = k, \kappa_1 = k_1, \dots, \kappa_n = k_n$ is a satisfying assignment of $\phi[len(s) / x]$. Suppose $\textsc{Encode}(\partialprog_1)$ returns $(\phi_1, x_1)$. Since $\partialprog[\kappa \annot k, \kappa_1 \annot k_1, \dots, \kappa_n \annot k_n]$ matches string $s$, we know  that $\partialprog_1[\kappa_1 \annot k_1, \dots, \kappa_n \annot k_n]$ must match $s_1, \dots, s_k$ where $s$ is the concatenation of $s_1, \dots, s_k$. By induction, we have that $\kappa_1 = k_1, \dots, \kappa_n = k_n$ is a satisfying assignment of $\phi_1[len(s_1) / x_1], \dots, \phi_1[len(s_n) / x_1]$. 
    Now we show that $\kappa = k, \kappa_1 = k_1, \dots, \kappa_n = k_n$ is a satisfying assignment of $\phi[len(s) / x]$ where $\phi$ is $\exists x_1, x'_1. \ (x \geq x_1 \kappa \wedge x \leq x'_1 \kappa )  \wedge \phi_1 \wedge \phi'_1[x'_1 / x_1] \wedge \phi_2$. This obviously holds (consider $x_1 = \emph{min}\{ len(s_1), \dots, len(s_k)\}$ and $x'_1 = \emph{max}\{ len(s_1), \dots, len(s_k) \}$). 
    \end{itemize}
    \end{itemize}
    \end{proof}

\section{Semantics for the Regex DSL}\label{appendix:dslsemantic}

    \footnotesize
        \[
            \begin{array}{l}
                \semantics{c} s = (s = c) \\ 
                \semantics{\texttt{StartsWith}( \prog )} s = \exists j. \  0 \le j < |s|. \ \semantics{\prog} s', \text{where } s'= s[0,j] \\
                \semantics{\texttt{EndsWith}( \prog )} s = \exists j. \ 0 \le j < |s|. \  \semantics{\prog} s', \text{where } s'=s[j,|s|-1]\\
                \semantics{\texttt{Contains}( \prog )} s = \exists i, j. \ 0 \le i \le j < |s|. \ \semantics{\prog} s', \text{where } s'=s[i,j] \\
                \semantics{\texttt{Not}( \prog )} s = \neg \semantics{\prog} s \\
                \semantics{\texttt{Optional}( \prog )} s = \ (s =  \epsilon \lor \semantics{\prog} s) \\
                \semantics{\texttt{Concat}( \prog_1, \prog_2 )} s = \exists j. \ 1 \leq j < |s|. \  \semantics{\prog_1} s_1 \land  \semantics{\prog_2} s_2, \\ \hspace{100pt}  \text{where } s_1=s[0,j], s_2=s[j+1,|s|-1] \\ 
                \semantics{\texttt{Or}( \prog_1, \prog_2)} s = \semantics{\prog_1} s \lor \semantics{\prog_2} s\\
                \semantics{\texttt{And}( \prog_1, \prog_2)} s = \semantics{\prog_1} s \land \semantics{\prog_2} s \\ 
                \semantics{\texttt{Repeat} ( \prog, k)} s = \left\{ 
                \begin{array}{lr}
                    \hspace{-5pt}
                    \begin{array}{lr}
                        \hspace{12pt}
                        \vspace{5pt}
                        \semantics{\prog} s  \ \ \ \ \ \ \ k = 1 \\
                        \exists j. \ 1 \le j < |s|. \  \semantics{\prog}s_1 \land \semantics{\texttt{Repeat} (\prog, k - 1)} s_2, \\ \hspace{10pt} s_1 = s[0, j], s_2=s[j+1, |s|-1] \ \ \ \ \ \ \ \text{otherwise}\\ 
                    \end{array}
                \end{array}
                \right. \\
                \semantics{\texttt{RepeatRange} (\prog, k_1, k_2)} s = \bigvee_{k_1 \leq k \leq k_2} \semantics{\texttt{Repeat} ( \prog, k)}s \\
                \semantics{\texttt{RepeatAtLeast} (\prog, k_1)} s = \bigvee_{k_1 \leq k \leq \infty} \semantics{\texttt{Repeat} (\prog, k)} s \\
                \semantics{\texttt{KleeneStar}( \prog )} s = \  (s = \epsilon) \lor \bigvee_{1 \leq k \leq \infty} \semantics{\texttt{Repeat} (\prog, k)} s \\ 
            \end{array}
        \]
\normalsize

\section{DeepRegex Data Set Details and Benchmark Collection Procedure}\label{appendix:drcollection}

In this section, we provide details about the \deepregex  dataset. Originally, \deepregex\ authors collected this dataset using the following methodology: 
First, they programmatically generate regular expressions and the corresponding synthetic natural language descriptions using a synchronous context-free grammar. Then, they ask Amazon Mechanical Turkers to paraphrase the synthetic English description in a way that sounds more natural \cite{wang2015building}. Using this methodology, they collect a total of 10,000 benchmarks consisting of both a natural language description and the corresponding regex.

However, for our purposes, there are three issues with the original \deepregex\ data set. First, since \deepregex\ does not utilize examples, these benchmarks do not contain any positive/negative string examples for the target regex. Second, the data set is quite noisy: for many of the benchmarks, the regex does not match the  description due to errors introduced during paraphrasing. Third, since the target regexes are randomly generated, most benchmarks are not very representative of string matching tasks that arise in the real world. For example, for approximately $1,400$ of the $10,000$ benchmarks, the generated regex actually corresponds to the empty language.

For the reasons explained above, we could not use the \deepregex\ data set as is for our purposes; however, we were able to adapt it and construct a suitable data set of $200$ benchmarks using the following methodology. First, we  removed all regexes that do not accept any strings. 
While this modification still does not guarantee that the resulting data set is completely representative of real-world tasks, 
it eliminates benchmarks that are completely unrealistic. 
Then, among the remaining benchmarks, we randomly sampled $800$ tasks and asked people at our institution to provide examples that they think best describe the desired task by \emph{only} looking at their English descriptions. In particular, we asked the users to provide \emph{up to}  $7$ (and no less than $2$) positive and $7$ (and also no less than $2$) negative examples for each benchmark. 

This process yielded $800$ benchmarks consisting of a natural language description, a target regex, and a set of positive/negative string examples. However, because the annotators did not see the ground truth regex, the labeled examples may not be consistent with it. If a benchmark had more than 3 incorrect examples\footnote{To clarify,  ``incorrect'' means that a negative example provided by the user  is accepted by the regex  or a positive example is not accepted by it}, we assume it is poorly paraphrased and discard it. Otherwise, if there are two or fewer incorrect examples, we simply discard the bad examples. We believe this also removes noise in the \deepregex\ data set and helps select those benchmarks whose natural language description, examples, and target regex are all compatible. 
Using this methodology, we managed to create a data set of $200$ benchmarks, consisting of the natural language description, 4-14 positive/negative examples, and a target regex.  On average, each benchmark contains 4 positive examples and 5 negative examples.

The \deepregex\ data set is included in the \emph{deepregex} folder of the supplementary materials.

\section{StackOverflow Data Set Details and Benchmark Collection Procedure}

We collected our StackOverflow data set using the following procedure. First,  we searched StackOverflow posts using keywords such as ``regex'', ``regular expression'', ``text validation'', ``password validation'' and retained all posts that satisfy the following criteria: (1) The post must contain both an English description of the task as well as positive and negative examples; (2) All the information relevant to the benchmark (i.e. examples, description and the solution) must be consistent with each other.  Using this methodology, we obtained a total of 122 benchmarks covering several categories of tasks, such as number matching, phone number matching, password validation and etc.

Since the original StackOverflow posts are quite noisy, we also pre-process these 122 benchmarks using the following methodology:  First, since many posts contain a description of the user's attempted solution, we removed such irrelevant parts of the post (e.g., ``I tried this regex but it's not working''). We also fixed typos in the English description and   added quotations around constants -- e.g., if the question text says ``write a regex for strings starting with .'', we would put the dot symbol in quotation marks. In addition, some of the benchmarks   involve visual formatting (e.g., ``key = value'') that cannot be parsed using NLP techniques. In such cases, we parse the visual format into a sketch outline and parse the description with respect to each part of the visual format independently.

We also write the ground truth for each benchmark in our DSL in order to check equivalence between the ground truth regex and the synthesized regex. The reason that we cannot do the opposite is that the library we used to check regex equivalence, \textsc{Automaton} \cite{automaton}, does not accept some constructs that show up in the standard regex, such as lookahead, while \textsc{Automaton} can accept any regex that is in our DSL (which means it accepts operators such as {\tt And} and {\tt Not}). 

The StackOverflow data set is included in the \emph{stackoverflow} folder of the supplementary materials. 

\section{Training for Each Data Set}

In order to train our semantic parser on labeled training data, we need to generate sketch labels for a given natural language description. 
For the \deepregex\ data set, we generate these  sketch labels from the target regex. Specifically, given a target regex $r$, we replace  the root operator \texttt{op} in $r$ with a hole whose components are \texttt{op}'s arguments. Following~\cite{deepregex}, we train Sempre on 6500 English sentences that are separate from 200 \deepregex\ benchmarks that we use to evaluate \toolname. While training, we set beam size to be 500 and batch size to be 1.

For the StackOverflow data set, we manually write sketch labels in a way that mimics the structure of the English utterance. For example, consider the sentence ``\emph{the input  box should accept only if either first 2 letters alpha +6 numeric or 8 numeric}''. 
The manually-written  h-sketch for this utterance is 
\footnotesize
$
    \texttt{Or} \big( \hole\{\texttt{Repeat}(\texttt{<let>,2}),\texttt{Repeat}(\texttt{<num>,6})\},\hole\{\texttt{Repeat}(\texttt{<num>,8})\} \big)
$
\normalsize
which contains  key building blocks like \texttt{Repeat}(\texttt{<let>,2}) of the target regex and indicates that the top-level construct is an \texttt{Or}. To train Sempre, we use 5-fold cross validation by dividing the data set into 5 non-overlapping folds and train on 4 folds while testing on the left-out fold. This procedure ensures that we never train on test data. For each fold, we train for 5 epochs, utilizing a beam size of 500 and a batch size of 1.

\section{User Study Procedure}\label{appendix:user}
 All benchmarks in our user study are randomly sampled from the StackOverflow data set.  We started 
 each user study session by first describing the task that the participants need to accomplish. In particular, the goal of the participants is to solve 6 regex-related tasks, and they are asked to solve three of these tasks with the help of \toolname\ and the remaining 3 regex tasks without  \toolname. 
 In order to minimize the effect knowledge transfer, we randomize whether a  participant is given access to \toolname\ for the first 3 tasks or the latter 3 examples.

After explaining the task, we next provided the participants with  a ``cheat-sheet'' that includes both standard regular expression syntax as well as our DSL. This cheat-sheet also further illustrates the semantics of each operator using  positive and negative examples. Each participant  was given 5 minutes to look over the cheat-sheet and familiarize themselves with this syntax. Afterwards, we gave participants a short demo (approximately 10 minutes) illustrating how one can use \toolname\ to generate a sample regex. The demo shows a simple manually-crafted regex task that is  not taken from the StackOverflow benchmarks. 

Once the participants understood the procedure, we asked them to complete the assigned tasks. In the setting involving \toolname, the workflow is similar to how we run the interactive \toolname\ presented in the evaluation section. Specifically, given the natural language description and examples, \toolname\ returns the top-3 synthesized regexes to the participant. Then, the user can either choose one of the returned regexes as the solution (if they think it is the intended one), or enter two more new examples. The participants are allowed to repeat this process as many times as they like within the given time limit. The users can also also provide a solution directly if they feel that they don't need another iteration of \toolname. For the setting not involving \toolname, we allow the participants to use Internet and the ``cheat-sheet'' that we provided at the beginning of the session as references. In particular, the participants are allowed to use \emph{any other resource} of their choice as long as they do not search for this particular regex task. In both settings, the participants have a total of $15$ minutes to finish the $3$ given regex tasks, and we do not restrict the time that they spend on each individual task. 

At the end of the session, we went over the participants' solutions and collected data on how many of the tasks they  successfully solved (i.e., provide a solution that matches the ground truth) with and without \toolname.


\end{appendices}